\documentclass[12pt]{amsart}
\usepackage[colorlinks,linkcolor=red,anchorcolor=blue,citecolor=green]{hyperref}
\usepackage{amssymb,amsthm,amsmath,amsfonts}
\usepackage[numbers,sort&compress]{natbib}

\usepackage[margin=1in]{geometry}
\usepackage{color}
\usepackage{bbm}
\usepackage{graphicx}
\usepackage{tikz}

\newcommand{\GL}{\mathrm{GL}}
\newcommand{\SL}{\mathrm{SL}}

\def\rL{\mathrm{L}}
\def\rC{\mathrm{C}}
\def\rH{\mathrm{H}}

\DeclareMathOperator{\im}{Im}
\DeclareMathOperator{\re}{Re}

\DeclareMathOperator*{\slim}{s--lim}

\def\a{\alpha}

\newcommand{\one}{\mathbf{1}}

\newcommand{\ba}{\overline{A}}

\newcommand{\cR}{\mathcal R}
\newcommand{\cRR}{{\mathcal{RR}}}

\def\be{\begin{equation}}
\def\ee{\end{equation}}

\def\ba{{\begin{align}}}
\def\ea{{\end{align}}}

\def\bm{\begin{matrix}}
\def\em{\end{matrix}}

\def\a{{\alpha}}

\def\r{\right}
\def\l{\left}

\def\SL{{\mathrm{SL}}}

\def\0{{\mathbf 0}}

\def\dc{\mathrm{DC}}

\def\cal{\mathcal}

\def\wt{\widetilde}

\theoremstyle{plain}
\newtheorem{theorem}{\bf Theorem}[section]
\newtheorem{lemma}[theorem]{\bf Lemma}

\newtheorem{prop}[theorem]{\bf Proposition}
\newtheorem{cor}[theorem]{\bf Corollary}

\theoremstyle{remark}
\newtheorem{remark}[theorem]{\bf Remark}
\newtheorem{defi}[theorem]{\bf Definition}

\numberwithin{equation}{section}

\def\<{\langle}
\def\>{\rangle}

\def\ssm{\smallsetminus}

\renewcommand{\setminus}{\ssm}

\DeclareMathOperator*{\esssup}{{ess--sup}}
\def\dom{\operatorname{Dom}}

\newcommand{\dist}{\operatorname{dist}}

\newcommand{\supp}{\operatorname{supp}}

\newcommand{\KK}{{\cal K}}

\newcommand{\RR}{{\cal R}}

\newcommand{\C}{{\mathbb C}}

\newcommand{\Q}{{\mathbb Q}}
\newcommand{\R}{{\mathbb R}}
\newcommand{\T}{{\mathbb T}}

\newcommand{\Z}{{\mathbb Z}}

\def\ac{{\mathrm{ac}}}
\def\pp{{\mathrm{pp}}}

\def\B0{{\bold{0}}}
\renewcommand{\phi}{\varphi}

\newcommand{\Ran}{\operatorname{Ran}}

\def\beq{\begin{equation}}
\def\eeq{\end{equation}}

\newcommand{\comm}[1]{}
\newcommand{\comment}[1]{}

\begin{document}

\title[]{Ballistic transport for one-dimensional quasiperiodic Schr\"odinger operators}
\author{Lingrui Ge}
\address{
Department of Mathematics, University of Califoria Irvine, 340 Rowland Hall, Irvine CA, 92697--3875, United States of America}
\email{lingruig@uci.edu}

\author[I. Kachkovskiy]{Ilya Kachkovskiy}
\address{Department of Mathematics,
	Michigan State University,
	Wells Hall, 619 Red Cedar Road,
	East Lansing, MI 48824,
	United States of America}
\email{ikachkov@msu.edu}

\maketitle
\begin{abstract}
In this paper, we show that one-dimensional discrete multi-frequency quasiperiodic Schr\"odinger operators with smooth potentials demonstrate ballistic motion on the set of energies on which the corresponding Schr\"odinger cocycles are smoothly reducible to constant rotations. The proof is performed by establishing a local version of strong ballistic transport on an exhausting sequence of subsets on which reducibility can be achieved by a conjugation uniformly bounded in the $\rC^{\ell}$-norm. We also establish global strong ballistic transport under an additional integral condition on the norms of conjugation matrices. The latter condition is quite mild and is satisfied in many known examples.
\end{abstract}

\section{Introduction}
\subsection{Types of ballistic motion}Let $H$ be a discrete Schr\"odinger operator on $\ell^2(\Z)$:
\beq
\label{h_def_abstract}
(H\psi)(n)=\psi(n-1)+\psi(n+1)+V_n\psi(n),\ \ n\in\Z.
\eeq 
where $\{V_n\}_{n\in \Z}$ is a sequence of real numbers (the potential). The operator $H$ is a Hamiltonian of a single quantum particle with wave function $\psi\colon \Z\to\C$, whose time evolution is described by the {\it time-dependent Schr\"odinger equation}:
\begin{equation}\label{schev}
i\frac{\partial\psi}{\partial t}=H\psi, \ \ \psi(0)\in\ell^2(\Z).
\end{equation}
Using the spectral theorem, one may explicitly solve \eqref{schev} via
\begin{equation}\label{time}
\psi(t)=e^{-itH}\psi(0).
\end{equation}
Let $B$ be a self-adjoint operator associated to an observable quantity. The {\it Heisenberg evolution} of $B$ is described by
$$
B(T)=e^{iTH}Be^{-iT H}.
$$
In the present paper, we will be interested in spatial transport properties of a quantum particle on the lattice $\Z$. The relevant observable quantity is the {\it position operator}
$$
(X\psi)(n):=n\psi(n),\ \ n\in\Z,
$$
which is an unbounded self-adjoint operator with the natural domain of definition
$$
\dom X=\{\psi\in\ell^2(\Z)\colon \sum_{n\in\Z}|n|^2|\psi(n)|^2<+\infty\}.
$$
One can check by direct calculation that the Heisenberg evolution of the position operator can be expressed in the following form:
\begin{equation}\label{evolution}
X(T):=e^{iTH}Xe^{-iTH}=X+\int_0^T e^{itH}Ae^{-itH}\,dt,\quad T\in \R,
\end{equation}
where
\begin{equation}\label{A_def}
A\psi(n)=i(\psi(n+1)-\psi(n-1))
\end{equation}
is sometimes known as the {\it current operator} (a tight binding analogue of the gradient operator $i\nabla$). 
Since $A$ is bounded, \eqref{evolution} implies that $X=X(0)$ and $X(T)$ have the same domain. We will be interested in the phenomenon of {\it ballistic motion}, which states, informally, that the position of the particle grows linearly with time (``$X(T)\approx T$''). More precisely, we will address the following limits:
\begin{equation}
\label{X_lim}
\lim\limits_{T\to +\infty}\frac{1}{T}X(T)\psi_0,
\end{equation}
where, initially, $\psi_0\in \dom X$. One can consider the limit \eqref{X_lim}, if it exists, as the ``asymptotic velocity'' of the state $\psi_0$ at infinite time. The {\it asymptotic velocity operator} can therefore be defined by
\begin{equation}\label{Q_def}
Q=\slim\limits_{T\to +\infty}\frac{1}{T} X(T)=\slim\limits_{T\to +\infty}\frac{1}{T}\int_0^T e^{itH}Ae^{-itH}\,dt. 
\end{equation}
The first limit can only be considered on $\dom X$, but, since the term $\frac{1}{T}X(0)$ of $\eqref{evolution}$ disappears as $T\to \infty$, it is natural to drop it from consideration. We say that a Schr\"odinger operator $H$ demonstrates {\it strong ballistic transport}, if the strong limit in the right hand side of \eqref{Q_def} exists, is defined on the whole $\ell^2(\Z)$ and, moreover, $\ker Q=\{0\}$.

An immediate consequence of \eqref{evolution} and \eqref{Q_def} is  that all of the moments of the position operator
grow ballistically in time. More specifically, for any $p\ge 2$ and $0\neq \psi_0\in \dom |X|^p$, we have
\begin{equation}\label{ballistic}
\lim\limits_{T\rightarrow +\infty}T^{-p}\<|X(T)|^p\psi_0,\psi_0\>=\<|Q|^p\psi_0,\psi_0\>>0.
\end{equation}
If we take $p=2$, \eqref{ballistic} immediately implies that $H$ has {\it ballistic motion}. More precisely, we say $H$ has ballistic motion if 
\begin{equation}\label{ballistic motion}
\liminf \limits_{T\rightarrow +\infty}T^{-2}\<|X(T)|^2\psi_0,\psi_0\>>0,\quad \psi_0\in \dom X,\quad \psi_0\neq 0.
\end{equation}
Note that \eqref{ballistic motion} is weaker than \eqref{ballistic} with $p=2$. One can also consider \eqref{ballistic motion} for $p\neq 2$.

Ballistic motion is one of the examples of wave packet spreading, which indicates absence of localization. 
The fundamental work in this aspect is the RAGE Theorem \cite{CFKS} which states for a Schr\"odinger operator $H$, if $\psi\in \ell^2_{\mathrm{c}}(H)$, then for any $N>0$,
\begin{equation}\label{rage}
{\lim\limits_{T\rightarrow+\infty}\frac{1}{T}\int_0^T\sum\limits_{|n|\leq N}|\langle\delta_n, e^{iTH}\psi\rangle|^2 dt=0,}
\end{equation}
where 
$$
\ell^2_{\mathrm{c}}(H)=\{\psi\in\ell^2(\Z):\mu_{\psi}=\mu_{\psi,\mathrm{c}}\}=\ell^2_{\pp}(H)^{\perp}
$$
is the subspace corresponding to the continuous spectrum of $H$. In other words, a wavepacket in the  continuous subspace of $H$ will spend most of the time outside of any fixed compact subset of $\Z$. In the case of the absolutely continuous subspace, \eqref{rage} can be further improved to a version that does not involve time averaging:
\begin{equation}\label{rage_ac}
{\lim\limits_{T\rightarrow+\infty}\sum\limits_{|n|\leq N}|\langle\delta_n, e^{iTH}\psi\rangle|^2 dt=0,}\quad  \psi\in \ell^2_{\ac}(H).
\end{equation}
Both \eqref{rage} and \eqref{rage_ac} imply the following growth conditions on the moments: 
\beq
\label{eq_rage}
\lim\limits_{T\to +\infty}\frac{1}{T}\int_0^T\<|X(t)|^p\psi_0,\psi_0\>\,dt=+\infty,\quad 0\neq \psi\in \ell^2_{\mathrm{c}}(H).
\eeq
\beq
\label{eq_rage_ac}
\lim\limits_{T\to +\infty}\<|X(T)|^p\psi_0,\psi_0\>=+\infty,\quad 0\neq \psi\in \ell^2_{\ac}(H).
\eeq
However, it is harder to estimate the exact rate of growth. In fact, this rate can be related to the Hausdorff dimension of the spectrum and spectral measures of $H$, see \cite{Last}. For the absolutely continuous case, the Guarneri--Combes--Last Theorem \cite{Last} states that, for any $p\ge 2$,
\begin{equation}\label{timeaverage}
\liminf \limits_{T\rightarrow +\infty}\frac{1}{T^{p+1}}\int_0^T\<|X(t)|^p\psi_0,\psi_0\>\,dt>0,\quad \forall \psi\in \dom |X|^p,\quad 0\neq \psi\in \ell^2_{\ac}(H).
\end{equation}
One can compare the above versions of transport as follows:
\begin{align}
\label{eq_diagram}
\begin{split}	
\textrm{existence of }\eqref{Q_def}\textrm{ with trivial kernel}\Rightarrow\eqref{ballistic}\Rightarrow&\,\eqref{ballistic motion}\textrm{ for all }p\Rightarrow\eqref{timeaverage}\Rightarrow  \eqref{eq_rage}\\
&\,\,\,\,\Downarrow\\
&\eqref{eq_rage_ac}.
\end{split}
\end{align}
Thus, strong ballistic transport (as defined in \eqref{Q_def}) can be viewed as the strongest version of ballistic motion. Note that, since the operator $A$ is bounded, \eqref{evolution} implies an elementary {\it ballistic upper bound} on the wave packet spreading. In other words, no transport can be stronger than ballistic.

In general, ballistic transport is not expected on any spectra other than purely absolutely continuous. In particular, it was shown in \cite{Simon_ballistic} that point spectrum cannot support any ballistic motion. However, one can still expect it after restricting the operator to a subspace that supports purely absolutely continuous spectrum. In this regard, we will need a version of the above definition that would be local in energy. Let $\KK\subset \R$ be a Borel subset. We will say that $H$ has {\it strong ballistic transport on $\KK$} if there exists a self-adjoint operator $Q$ such that
\beq
\label{eq_ballistic_K_def}
\slim\limits_{T\to +\infty}\frac{1}{T}\int_0^T e^{itH}\one_{\KK}(H)A\one_{\KK}(H)e^{-itH}\,dt=\one_{\KK}Q\one_{\KK}
\eeq
and $\ker Q=\Ran(\one_{\KK})^{\perp}$. While we will be able to establish \eqref{eq_ballistic_K_def} in a range of situations, a significant gain in generality can be achieved by slightly relaxing the above definition. We will say that $H$ has {\it local ballistic transport on $\KK$} if there exists a self-adjoint operator $Q$ and a sequence of Borel subsets $\{\KK_j\}_{j=1}^{\infty}$ such that $\KK=\cup_j \KK_j$ and $H$ satisfies \eqref{eq_ballistic_K_def} on each $\KK_j$. As a part of the definition, we require that $Q$ is the same operator for all $\KK_j$, and $\ker Q(\KK)=\Ran(\one_{\KK})^{\perp}$. For the purpose of the diagram \eqref{eq_diagram}, local ballistic transport implies lower bounds on wavepacket spreading just as good as strong ballistic transport. More precisely, if $\psi\in \Ran\one_{\KK}(H)$, then, for large $j$, we have
\beq
\label{eq_ballistic_conclusion}
\frac{1}{T}\int_0^T e^{itH}Ae^{-itH}\psi\,dt=\one_{\KK_j}Q\psi+\psi^{\perp}(T)+o(1),
\eeq
where $\psi^{\perp}(T)$ is orthogonal to $\one_{\KK_j}Q\psi$, and hence can only increase the norm. Note that we are using the right hand side of \eqref{X_lim} instead of $\frac{1}{T}X(T)\psi$, since we cannot guarantee that the intersection $\Ran\one_{\KK}(H)\cap \dom X$ is large enough. However, if $\psi\in \dom X$ is sufficiently close to $\Ran \one_{\KK}(H)$ (for example, $\|(1-\one_{\KK}(H))\psi\|< \frac12\|\one_{\KK} Q\psi\|$), then \eqref{eq_ballistic_conclusion} implies a ballistic lower bound on $\|X(T)\psi\|$. The set of such $\psi$ is dense in $\Ran \one_{\KK}(H)$. The difference between \eqref{eq_ballistic_K_def} and \eqref{eq_ballistic_conclusion} is that the latter may have a non-trivial ``tail'' which stays within the range of $\one_{\KK}(H)$, but eventually escapes any $\Ran(\one_{\KK_j}(H))$ with finite $j$. However, this tail can only strengthen the ballistic lower bound. As a consequence, local ballistic transport still implies \eqref{ballistic} and \eqref{ballistic motion}.

Unlike \eqref{timeaverage} and \eqref{eq_rage_ac}, we are not aware of any results of the form \eqref{Q_def}--\eqref{ballistic motion} for general Schr\"odinger operators with absolutely continuous spectra.\footnote{Except for potentials decaying on infinity, where one can obtain these bounds using scattering theory. In general, ballistic transport is expected to be stable under decaying perturbations. We do not go into the details in the present paper.} Instead, all known results only apply to potentials of special structure. First results of this type were obtained in \cite{AschKnauf} for periodic operators in the continuum. Later, a tight binding analogue was obtained in \cite{DamanikLY} for discrete periodic Jacobi matrices, motivated by applications to XY spin chains. See also related paper \cite{DLLY} about anomalous (non-ballistic) transport for Fibonacci-type operators with singular continuous spectra. The limit-periodic case was studied in \cite{fillman} where an analogue of \eqref{X_lim} was proved by periodic approximations.

\subsection{Quasiperiodic operators }The next natural class of operators with absolutely continuous spectra, where one can expect ballistic motion/ballistic transport, is {\it quasiperiodic Schr\"odinger operators}, which will be the subject of the present paper. Let $v\in C^{s}(\T^d;\R)$ be a smooth function. We will identify $\Z^d$-periodic functions on $\R^d$ with functions on $\T^d$. Let also $\alpha\in \R^d$ be a {\it frequency vector}. We will always assume that $\{1,\alpha_1,\ldots,\alpha_d\}$ are independent over $\Q$. An {\it smooth multi-frequency quasiperiodic Schr\"odinger operator} is an operator of the form
\beq
\label{h_def}
(H_{x}\psi)(n)=\psi(n-1)+\psi(n+1)+v(x+n\alpha)\psi(n),\ \ n\in\Z.
\eeq
Here $x\in \T^d$ is the quasiperiodic phase, and one usually considers the whole family $\{H_x\}_{x\in \T^d}$.

Quasiperiodic operators \eqref{h_def} with small analytic potentials $v$ are often known to have purely absolutely continuous spectra, see \cite{E92,Amor,aj,BJ,AK,AFK}. In \cite{Kach1}, it was shown that a large class of such operators (in all cited regimes, except for the Liouville case in \cite{AFK}) satisfies {\it $x$-averaged strong ballistic transport}. In other words, instead of \eqref{X_lim}, one has the following convergence statement in the direct integral space $\mathrm{L}^2(\T^d\times \Z)$:
$$
\slim_{T\to +\infty}\l(\frac{1}{T}\int_{\T^d}^{\oplus}X(x,T)\,dx\r)=\int_{\T^d}^{\oplus}Q(x)\,dx.
$$
where $X(x,T):=e^{iTH_x}Xe^{-iTH_x}$. The proof used the duality method based on \cite{jk2}. Like \cite{DamanikLY}, the work \cite{Kach1} was motivated by applications to the XY spin chains. The $x$-averaged version of ballistic transport implies existence of the limit \eqref{Q_def} on a subsequence of time scales for almost every $x$ and hence is sufficient for the conclusion on the XY spin chain. However, it does not imply any of the claims \eqref{Q_def}--\eqref{ballistic motion} in full. In the same year, a KAM-type approach was developed in \cite{zhao1} in order to obtain bounds of type \eqref{ballistic motion} in the perturbative setting. The advantage is that it works for all $x$ and does not require to take a subsequence of time scales. However, it falls short of establishing existence of \eqref{Q_def}. The KAM method of \cite{zhao1} was later expanded in \cite{zhaozhang} to treat the one-frequency Liouvillean case, by further weakening \eqref{ballistic motion} to a lower bound on some transport exponents.

\subsection{Outline of the approach} The goal of the present paper is to obtain a result which has the advantages of both \cite{Kach1} and \cite{zhao1}. One can see it as a refinement of either of the papers, however, the general line of the argument is closer to \cite{Kach1}. In the quasiperiodic case, one of the results of \cite{Kach1} is the calculation of the asymptotic velocity operator $Q(x)$, but, since it is only obtained on a sequence of time scales, one cannot exclude the possibility of large oscillations around the limiting value. Moreover, \cite{Kach1} predicts a possible mechanism of convergence: after applying duality, it becomes a procedure of diagonal truncation of an operator dual to \eqref{A_def} in the basis of the eigenvectors of the dual Hamiltonian with purely point spectrum. The convergence of the truncation is only obtained in the Fourier dual direct integral space $\rL^2(\T\times\Z^d)$, which is not enough to guarantee pointwise strong convergence in the original direct integral space. {\it A natural question arises: can extra information on the dual operator family improve the rate of convergence in \eqref{Q_def}? If yes, what kind of information can be used?}

In order to obtain a pointwise bound, we would like to replace $\rL^2(\T^d\times\Z)$ by $L^{\infty}(\T^d;\ell^2(\Z))$ or by $C(\T^d;\ell^2(\Z))$. One can try to obtain that by improving  convergence in the dual space: for example, to $\ell^1(\Z^d;\rL^2(\T))$. On any finite box, $\ell^1$ and $\ell^2$-norms are equivalent (with the constants depending on the size of the box). Therefore, one possible way of obtaining $\ell^1$-convergence would be to obtain a uniform $\ell^1$ bound on the tails. The latter can be achieved by investigating quantitative character of the localization for the dual model. For example, one can take advantage of {\it exponential dynamical localization in expectation} which has been obtained in \cite{JitoKru} and \cite{gyz} under some assumptions. Along these lines, one can obtain the desired control on the tails, which would imply strong ballistic transport for almost every $x\in \T^d$. This approach was partially implemented in the preprint \cite{Kach2}, which is no longer intended for publication since the current paper supersedes it in several ways. 

The main results of the present paper are Theorem \ref{main} (the local result) and Theorem \ref{main2} (the global result). In Theorem \ref{main}, we state that the operators \eqref{h_def} have local ballistic transport on the set of energies on which the corresponding Schr\"odinger cocycles are $\rC^s$-reducible with $s>d$ (see Section 2.1 for precise definitions). We do not require any quantitative information on the conjugating matrices and do not care about Diophantine properties of the frequency vector. While the result falls short of the complete strong ballistic transport, most of its conclusions (such as ballistic motion) also hold, as described above. In Theorem \ref{main2}, we state that one can obtain strong ballistic transport under an additional integral condition on the norms of the conjugating matrices. Several known examples, including the settings of \cite{zhao1} and \cite{gyz}, satisfy this condition.

The proof of the local result is based on the following observation: suppose that $\cR$ is the set of energies under consideration, and 
$$
\KK_1\subset\KK_2\subset\ldots\subset \RR
$$ 
is a sequence of Borel subsets such that $\RR\setminus(\cup_j \KK_j)$ has zero spectral measure with respect to $H_x$. Then, it is sufficient to check that the limit \eqref{eq_ballistic_K_def} exists on each $\KK_j$. The main problem in obtaining ``nice'' localization bounds for the dual model is the fact that regularity of the conjugation matrices (Bloch waves) is not uniform in the energy and depends on the Diophantine properties of the rotation number. Quantitative estimates, such as in \cite{gyz}, can be quite delicate. On the other hand, if one is allowed to restrict to a subset of energies, we can get, basically, as good control of the localization parameters as desired. In particular, we can get a ridiculously strong version of uniform localization, which is not even remotely available on the whole spectrum. As expected, the constants will get worse as one increases the set of energies under consideration. Since we only need $\ell^1$ control of the tails, we also do not require Anderson (exponential) localization, and are satisfied with polynomial decay of eigenfunctions, which allows us to consider smooth potentials rather than analytic. The idea of restricting to an exhausting subset of energies/rotation numbers while maintaining control on the regularity is not unlike the argument in \cite{gy}.

The global result is somewhat more delicate. While we cannot expect any uniform reducibility bounds, the desired bound still contains an integral in $\theta$, and hence, just as in the proofs of dynamical localization, one can hope for a quantitative result ``in expectation''. Using a variant of the covariant representation for the eigenfunctions of the dual operator by duality such as in \cite{jk2}, we reduce the integral 
\beq
\label{eq_dynamical}
\int_{\T}|\<\delta_{m},e^{it L_{\theta}}\delta_n\>|\,d\theta
\eeq
that appears in the proof of dynamical localization, to a convolution-type bound on the eigenfunctions which, in turn, can be controlled in terms of $\rC^s$ or Sobolev norms of the conjugating matrices, averaged over the rotation number. Unfortunately, in order to obtain better bounds, we would ideally want to estimate a different, smaller integral
$$
\int_{\T}|\<\delta_{m},e^{it L_{\theta}}\delta_n\>|^r\,d\theta,\quad r>1,
$$
and we were not able to find any way to take advantage of $r>1$, which actually appears in our desired bounds. Still, by taking some losses, we were able to obtain a bound by a series of convolution-type estimates for expressions of the form \eqref{eq_dynamical}. As a result, in the global theorem, the smoothness requirement is of the form $\rC^s$ with $s>5d/2$, rather than $s>d$ as in the local result. Still, our integral condition is satisfied by a large margin in the models where exponential dynamical localization is obtained such as \cite{gyz}. It is also easy to reformulate our global result as a conditional one: for example, strong ballistic transport will hold on $\KK$ if we assume $\rC^1$-reducibility on $\KK$ (without any quantitative control) and power law dynamical localization on $\KK$:
$$
\int_{\T}|\<\delta_{m},\one_{\KK}(L_{\theta})e^{it L_{\theta}}\delta_n\>|\,d\theta\le \frac{C}{(1+|m-n|)^s},\quad s>4d,
$$
which is weaker than, say, exponential dynamical localization in expectation.

In both cases, the stated arguments would only imply the corresponding version of ballistic transport for almost every $x\in\T^d$, since the duality ignores measure zero subsets of phases. In the case of the dual operator, this can be a real issue: for example, one cannot expect localization for all $\theta\in \T$ \cite{js}. However, quantities related to the absolutely continuous spectrum are known to be more phase stable. We were able to recover continuity in $x$ by comparing the pre-limit expressions in the definition of $Q(x)$ and the alternative definition of $Q(x)$ and showing that they are both uniformly continuous in the strong operator topology. In the latter case, we had to use quantitative continuity of the absolutely continuous spectral measures discussed in Section 5. As stated, one can only obtain it in the setting of local ballistic transport, since one has to restrict the operator to one of the subsets $\KK_j$. However, that particular part survives after passing to the union of $\KK_j$, and thus is also applicable to the global case.

\subsection{Acknowledgments} 
Ge is partially supported by the NSF grant DMS--1901462, and Kachkovskiy is currently supported by NSF DMS--1846114.

Both authors would like to thank S. Jitomirskaya for facilitating their collaboration and for comments on the manuscript.

\section{Preliminaries and statements of the results}
\subsection{Schr\"odinger cocycles and reducibility}
Let $A \in \rC^s(\T^d;{\SL}(2,\R))$, and consider a frequency vector $\alpha\in\R^d$ such that $\{1,\alpha_1,\ldots,\alpha_d\}$ are independent over $\Q$. By definition, a {\it quasiperiodic $\rC^s$-smooth $\SL(2,\R)$-cocycle} is a map
$$
(\alpha,A)\colon \begin{cases}
\T^d \times \C^2 \to  \T^d \times \C^2;\\
(x,v) \mapsto (x+\alpha,A(x)v).
\end{cases}
$$
The iterates of $(\alpha,A)$ are of the form $(\alpha,A)^n=(n\alpha,  A_n)$, where
$$
{A}_n(x):=
\begin{cases}
A(x+(n-1)\alpha) \cdots A(x+\alpha) A(x),  & n\geq 0;\\
A^{-1}(x+n\alpha) A^{-1}(x+(n+1)\alpha) \cdots A^{-1}(x-\alpha), & n<0.
\end{cases}
$$
We will usually simply call the above maps cocycles. Similarly, one can talk about $\SL(2,\C)$-cocycles. The {\it Lyapunov exponent} of the cocycle $(\alpha,A)$ is defined by
$$
\displaystyle
L(\alpha,A):=\lim\limits_{n\to \infty} \frac{1}{n} \int_{\T^d} \ln \|A_n(x)\| dx.
$$

A cocycle $(\alpha,A)$ is called {\it uniformly hyperbolic} if, for every $x \in \T^d$, there exists a continuous splitting $\C^2=E^s(x)\oplus E^u(x)$ such that for every $n \geq 0$,
$$
\begin{array}{rl}
|{A}_n(x) \, v| \leq C e^{-cn}|v|, &  v \in E^s(x),\\[1mm]
|{A}_n(x)^{-1}   v| \leq C e^{-cn}|v|, &  v \in E^u(x+n\alpha),
\end{array}
$$
for some constants $C,c>0$.
This splitting is invariant under the dynamics, i.e.,
$$A(x) E^{s}(x)=E^{s}(x+\alpha), \quad A(x) E^{u}(x)=E^{u}(x+\alpha),\quad  \forall \  x \in \T^d.$$

Assume that $A \in C^0(\T^d;{\rm SL}(2, \R))$ is homotopic to the identity. It induces the projective skew-product $F_A\colon \T^d \times \mathbb{S}^1 \to \T^d \times \mathbb{S}^1$ with
$$
F_A(x,w):=\left(x+\a,\, \frac{A(x) \cdot w}{|A(x) \cdot w|}\right).
$$
In other words, $F_A\colon \T^d\times\T\to \T^d\times\T$ can be expressed as $(x,y)\mapsto (x+\alpha,y+\phi_x(y))$, where $\phi_x\colon \R\to \R$ is a $1$-periodic continuous function (defined modulo translations by integers on both copies of $\R$). Let $\mu$ be any probability measure on $\T^d \times \T$ invariant under ${F}_A$ and whose projection onto the coordinate $x$ is given by the Lebesgue measure. The number
\beq
\label{eq_rho_def}
\rho(\alpha,A):=\int_{\T^d \times \T} \phi_x(y)\ d\mu(x,y) \ {\rm mod} \ \Z
\eeq 
depends  neither on the lift $\phi$ nor on the measure $\mu$, and is called the \textit{fibered rotation number} of $(\alpha,A)$ (see \cite{H,JM} for more details; see also \cite[Appendix A]{AK} for a detailed exposition). Let $R_{\theta}$ denote the rotation matrix
\beq
\label{eq_rotation_def}
R_\theta:=
\begin{pmatrix}
\cos2 \pi\theta & -\sin2\pi\theta\\
\sin2\pi\theta & \cos2\pi\theta
\end{pmatrix},\quad \theta\in \T.
\eeq
Any continuous map $A\colon \T^d\to{\rm SL}(2,\R)$ is homotopic to $x \mapsto R_{n\cdot x}$ for a unique $n\in\Z^d$. We call $n$ the {\it degree} of $A$ and denote it by $\deg A$.
The fibered rotation number is invariant under real conjugacies which are homotopic to the identity. More generally, if $(\alpha,A_1)$ is conjugated to $(\alpha, A_2)$, i.e., $B(x+\alpha)^{-1}A_1(x)B(x)=A_2(x)$, for some $B \colon \T^d\to{\rm PSL}(2,\R)$ with $\deg{B}=n$, then
\begin{equation}\label{rotation number}
\rho(\alpha, A_1)= \rho(\alpha, A_2)+n\cdot\alpha.
\end{equation}
A typical  example of a quasiperiodic cocycle is a \textit{Schr\"{o}dinger cocycle} $(\alpha,S_{E-v})$, where
$$
S_{E-v}(x):=
\begin{pmatrix}
E-v(x) & -1\\
1 & 0
\end{pmatrix},   \quad E\in\R.
$$
Any formal solution $\psi=\{\psi(n)\}_{n \in \Z}$ of the eigenvalue equation $H_x \psi=E\psi$, where $H_x$ is the quasiperiodic Schr\"odinger operator \eqref{h_def}
$$
(H_{x}\psi)(n)=\psi(n-1)+\psi(n+1)+v(x+n\alpha)\psi(n),\ \ n\in\Z,\,\, x\in \T^d,
$$
satisfies the following relation with $S_{E-v}(x)$:
$$
\begin{pmatrix}
\psi_{n+1}\\
\psi_n
\end{pmatrix}
= S_{E-v}(x+n\alpha) \begin{pmatrix}
\psi_{n}\\
\psi_{n-1}
\end{pmatrix},\quad \forall n \in \Z.
$$
It is well known that the spectrum $\sigma(H_{x})$, denoted by $\Sigma_{\alpha,v}$, is a compact subset of $\R$, independent of $x$ if $\{1,\alpha_1,\ldots,\alpha_d\}$ are rationally independent.
The spectral properties of $H_{x}$ and the dynamics of $(\alpha,S_{E-v})$ are related by the Johnson's theorem \cite{John}: $E\in \Sigma_{\alpha,v}$ if and only if $(\alpha,S_{E-v})$ is \textit{not} uniformly hyperbolic. Throughout the paper, we will use the notation $L(E)=L(\alpha,S_{E-v})$  and $\rho(E)=\rho(\alpha,S_{E-v})$ for brevity.

\subsection{Reducibility of quasiperiodic cocycles} We will only consider cocycles $(\alpha,A)$ with $\deg A=0$. A quasiperiodic $\rC^s$-cocycle $(\alpha,A)$ with $\{1,\alpha_1,\ldots,\alpha_d\}$ rationally independent is called {\it $\rC^s$-rotations reducible} if there exists $B\in \rC^s(\T^d;\SL(2,\R))$ and $\theta\in \rC^s(\T^d;\R)$  such that
\beq
\label{eq_reducible_def}
B(x+\alpha)^{-1}A(x)B(x)=R_{\theta(x)}.
\eeq
We will call a cocycle {\it reducible} if it is rotations reducible to a constant rotation. In this case, one can choose $\theta\equiv\rho(\alpha,A)$. For reducible cocycles, it will be more convenient diagonalize the rotation matrix and consider $B\in \rC^s(\T^d;\SL(2,\C))$ satisfying
\beq
B(x+\alpha)^{-1}A(x)B(x)=\begin{pmatrix}e^{2\pi i\rho(\alpha,A)}&0\\0&e^{-2\pi i \rho(\alpha,A)}\end{pmatrix}.
\eeq
Note that our use of the definition is more narrow than usual. More accurately, we should have used the wording ``reducible to a constant rotation''. Usually, one considers reducibility to a general constant matrix in the right hand side of \eqref{eq_reducible_def}.

Let $\{H_x\}_{x\in \T^d}$ be a quasiperiodic operator family, and $(\alpha,S_{E-v})$ be the corresponding Schr\"odinger cocycle. Define the following subset:
\begin{multline*}
\cR_{\alpha,v}^s=\{E\in \R\colon (\alpha,S_{E-v})\text{ is $\rC^s$-reducible}\}\\
\subset \cRR^s_{\alpha,v}=\{E\in \R\colon (\alpha,S_{E-v})\text{ is $\rC^s$-rotations reducible}\}.
\end{multline*}
We will sometimes drop the indices and simply use $\cR$ or $\cRR$, if the values of the indices are clear from the context.

From Shnol's theorem \cite{Schnol1,Schnol2,Schnol3,Schnol4}, it follows that $\cRR_{\alpha,v}^s\subset\Sigma_{\alpha,v}$. Moreover, subordinacy theory \cite{KP,JL_sub2,JL_sub3} implies that {\it the restriction of the spectral measure of $H_x$ into $\cRR_{\alpha,v}^s$ is purely absolutely continuous} for any $s\ge 0$. The same also holds for $\cR_{\alpha,v}^s$.

We will also need some conventions about normalizations of the cocycles in $\rL^2(\T^d)$. Let us rewrite the reducibility equation \eqref{eq_reducible_def_sec3} as
\begin{small}
$$
\begin{pmatrix}
(E-v(x))b_{11}(x)-b_{21}(x)& (E-v(x))b_{12}(x)-b_{22}(x)\\
b_{11}(x)&b_{12}(x)
\end{pmatrix}=\begin{pmatrix}
e^{2\pi i \theta}b_{11}(x+\alpha)&e^{-2\pi i \theta}b_{12}(x+\alpha)\\
e^{2\pi i \theta}b_{21}(x+\alpha)&e^{-2\pi i \theta}b_{22}(x+\alpha)
\end{pmatrix}.
$$
\end{small}
One can see that the columns of $B(x)$ are not intertwined, and one can multiply one column and divide another by the same constant without affecting the determinant. Note also that 
$\|b_{11}\|_{\rL^2(\T^d)}=\|b_{21}\|_{\rL^2(\T^d)}$, $\|b_{12}\|_{\rL^2(\T^d)}=\|b_{22}\|_{\rL^2(\T^d)}$. As a consequence, we can choose a constant so that the columns are ``balanced'':
$$
\|b_{11}\|_{\rL^2(\T^d)}=\|b_{21}\|_{\rL^2(\T^d)}=\|b_{12}\|_{\rL^2(\T^d)}=\|b_{22}\|_{\rL^2(\T^d)},
$$
without affecting the regularity of the matrix $B$ in the variable $x$. So, we would have
\beq
\label{eq_balanced}
\|B\|_{\rL^2(\T^d)}^2=4\|b_{ij}\|_{\rL^2(\T^d)}^2,\quad \forall i,j\in\{1,2\},
\eeq
where in the left hand side we are using the Hilbert--Schmidt matrix norm. In the statements of the main results, we will always assume that the conjugation matrix $B$ is balanced in the above sense. Also, we will not always require $\det B(x)=1$, but sometimes instead choose $B$ to be $\rL^2$-normalized (and balanced).

\subsection{Statements of the results}In order to formulate the main result, we will need the definition of the {\it density of states measure} of the operator family $\{H_x\}_{x\in \T^d}$: for a  Borel subset $B\subset \R$, define
\beq
\label{eq_ids_average}
N(B)=\int_{\T^d} \<\one_{B}(H_x)\delta_0,\delta_0\>\,dx.
\eeq
In other words, the density of states measure is the expectation value of the spectral measures. We also introduce the {\it integrated density of states} (denoted by the same symbol with a slight abuse of notation):
\beq
\label{eq_ids_def}
N(E):=N((-\infty,E])=N((-\infty,E)),
\eeq
the cumulative distribution function of the density of states measure. It is well known that $N$ is a non-decreasing continuous function of $E$. Clearly, if the spectral measures are absolutely continuous, then the IDS is also absolutely continuous (with respect to the Lebesgue measure). The IDS is related to the fibered rotation number defined above in \eqref{eq_rho_def} in the following way \cite{DS}:
$$
N(E)=1-2\rho(E).
$$
Let $\KK\subset \R$ be a Borel subset. The following function will be important:
\beq
\label{eq_g_def}
g_{\KK}(E)=\begin{cases}\frac{1}{\pi N'(E)},&\quad E\in \KK\\
0,&E\in \R\backslash\KK.	
\end{cases}
\eeq
Note that  $g_{\KK}(E)$ is well defined (Lebesgue) almost everywhere on $\KK\cap\Sigma_{\alpha,v}$. As a consequence, the operator $g_{\KK}(H_x)$ is well defined as long as $H_x$ has purely absolutely continuous spectrum on $\KK$.

For a (Borel) subset $\KK\subset \R$, denote by $\one_{\KK}(x)$ the indicator function of $\KK$. If $H$ is a self-adjoint operator on $\ell^2(\Z)$, denote by $H(\KK)$ the restriction of $H$ into the subspace $\Ran \one_{\KK}(H)\subset \ell^2(\Z)$. Here, $\one_{\KK}(H)$ is considered in the standard sense of functional calculus for self-adjoint operators. For the current operator $A$ defined in \eqref{A_def}, let
$$
A(x,\KK):=\one_{\KK}(H_x)A\one_{\KK}(H_x).
$$
In the case of $A(x,\KK)$, it is convenient not to restrict it into $\Ran \one_{\KK}(H_x)$ and instead let it have a zero block.

We are ready to formulate the first (local) main result of the paper.
\begin{theorem}\label{main} Let $\{H_x\}_{x\in \T^d}$ be a quasiperiodic operator family with $v\in \rC^s(\T^d;\R)$, $s>d$. Denote by $\cR$ the set of energies on which the corresponding Schr\"odinger cocycle is $\rC^s$-reducible. Then $H_x$ has local ballistic transport on $\cR$. In other words, there exists a representation $\cR=\cup_j\KK_j$ such that the following limit exist for all $x\in \T$ and all $\KK_j$:
\begin{equation*}
Q(x,\KK_j)=\slim\limits_{T\to +\infty}\frac{1}{T}\int_0^T e^{itH_x}A(x,\KK_j)e^{-itH_x}\,dt=g_{\KK_j}(H_x).
\end{equation*}
As a consequence,
$$
\ker{Q(x,\KK_j)}=(\Ran \one_{\KK_j}(H_x))^{\perp}.
$$
\end{theorem}
Theorem \ref{main} is ``soft'' and requires very little regularity. As a consequence, we can only get local bounds. Still, as stated in the Introduction, even these bounds imply ballistic motion such as in \cite{zhao1}. If $\cR$ has 

If one has some control over the dependence of $\|B\|_{\rC^s}$ in the $E$ variable, the result can be improved to ``true'' strong ballistic transport. Unfortunately, there is no hope in getting any kind of estimates that are uniform in energy, since regularity of the reducibility matrix depends on Diophantine properties of the rotation number (see, for example, \cite{gyz}). However, we can formulate a sufficient integral-type condition. We will say that $(\alpha,S_{v-E})$ is {\it $\rC^s$-reducible in expectation on $\KK$} if it's $\rC^s$-reducible for every $E\in \KK$, and there exists a choice of $\rL^2(\T^d)$-normalized conjugations $B(E;x)$ such that
\beq
\label{eq_integral_condition}
\int_{\KK}\|B(E;\cdot)\|_{\rC^s(\T^d)}^4\,d\rho(E)<+\infty.
\eeq
We can now formulate the second (global) main result.
\begin{theorem}\label{main2} Let $\{H_x\}_{x\in \T^d}$ be a quasiperiodic operator family whose cocycles are is $\rC^s$-reducible in expectation on $\KK$ for some $s>5d/2$. Then the family $\{H_x\}_{x\in \T^d}$ has strong ballistic transport on $\KK$. In other words,  the following limit exists for all $x\in \T$:
\begin{equation*}
Q(x,\KK)=\slim\limits_{T\to +\infty}\frac{1}{T}\int_0^T e^{itH_x}A(x,\KK)e^{-itH_x}\,dt=g_{\KK}(H_x).
\end{equation*}
As a consequence,
$$
\ker{Q(x,\KK)}=(\Ran \one_{\KK}(H_x))^{\perp}.
$$
\end{theorem}
We will also state a version of Theorem \ref{main2} in terms of the localization property of the dual operator
$$
(L_{\theta}\psi) (n)=\sum_{m\in \Z^d}\hat{v}_{n-m}\psi(m)+2\cos 2\pi (n\cdot\alpha+\theta)\psi(n),\quad n\in\Z^d.
$$
We will say that the family $\{L_{\theta}\}_{\theta\in \T}$ has {\it $s$-power law dynamical localization} (sPDL) on $\KK$, if the spectra of $L_{\theta}(\KK)$ are purely point for almost every $\theta\in\T$, and there are $C>0$, $s>0$ such that
$$
\int_{\T}|\<\delta_{m},\one_{\KK}(L_{\theta})e^{it L_{\theta}}\delta_n\>|\,d\theta\le \frac{C}{(1+|m-n|)^s}.
$$
The following is a corollary of the proof of Theorem \ref{main2}.
\begin{cor}
\label{cor_main3}
Let $\{H_x\}_{x\in \T^d}$ be a quasiperiodic operator family whose cocycles are $\rC^1$-reducible on $\KK$, and the dual family $\{L_{\theta}\}_{\theta\in\T}$ satisfies $s$-power law dynamical localization on $\KK$ with some $s>4d$. Then the family $\{H_x\}_{x\in \T^d}$ has strong ballistic transport on $\KK$.
\end{cor}

The assumptions of Theorems \ref{main} and/or \ref{main2} are satisfied for several different classes of operators. In order to formulate some of them, recall that a frequency vector $\alpha\in \R^d$ is called {\it Diophantine} (denoted $\alpha\in \dc_d(\gamma,\tau)$ for some $\gamma>0,\tau>d-1$) if
\begin{equation}
\label{dc_def}
\dist(k\cdot\alpha,\Z)\ge \gamma|k|^{-\tau},\quad \forall k\in \Z^d\backslash\{0\}.
\end{equation}
We will use the notation 
$$
\dc_d=\bigcup_{\gamma>0;\,\tau>d-1} \dc(\gamma,\tau).
$$ 
In the one-frequency case $\alpha\in \R\backslash\Q$, denote also
$$
\beta(\alpha)=\limsup_{k\to \infty}\frac{\ln q_{k+1}}{q_k},
$$
where $\frac{p_k}{q_k}\to \alpha$ are the continued fraction approximants. Note that $\alpha\in \dc_1$ implies $\beta(\alpha)=0$, but not vice versa. 
\begin{remark}
\label{remark_sobolev}
Condition \eqref{eq_integral_condition} is formulated in terms of $\rC^s$-norms in order to avoid overloading this section with terminology. In fact, in the proof we will use (weaker) Sobolev $\rH^s$-norms, since they behave better under some convolution-type operations appearing in the process.
\end{remark}
\subsection{Applications of Theorems \ref{main} and \ref{main2}} As stated earlier, Theorem \ref{main} falls in the middle between ballistic motion and strong ballistic transport. Its advantage is that it is applicable in a wide range of situations.
\begin{enumerate}
	\item Let $v\in C^{\omega}(\T;\R)$ be an analytic one-frequency potential, and $\beta(\alpha)=0$. Then there exists a Borel subset $\Sigma\subset \R$ such that $\Sigma$ supports the absolutely continuous components of the spectral measures of $H_x$ for all $x\in \T$, and the corresponding Schr\"odinger cocycle $S_{E-v}$ is analytically rotations reducible for all $E\in \Sigma$ due to \cite[Theorem 1.2]{AFK}. Since $\beta(\alpha)=0$, by solving the cohomological equation, one can improve rotations reducibility to reducibility for all $E\in\Sigma$. Thus, Theorem \ref{main} applies. One can state its conclusion in the following way: if $\{H_x\}_{x\in \T}$ is an analytic one-frequency quasiperiodic operator family with $\beta(\alpha)=0$ and $\Sigma$ does not support singular spectral measures of $H_x$, then $H_x$ has local ballistic transport and, as a consequence, has ballistic motion on $\Sigma$.
	\item In \cite{ck_zhao}, some of the results \cite{AK,AFK} were extended to the case of finitely smooth cocycles. As a consequence, the results from the previous case also extend to finitely differentiable potentials.
\item In \cite{ayz}, it was shown that almost Mathieu operators with potentials $v(x)=2\lambda\cos(2\pi x)$ with $\log\lambda<-\beta(\alpha)$ satisfy full measure analytic reducibility. As a consequence, they also satisfy local ballistic transport (and hence ballistic motion) on the whole spectrum. The corresponding quantitative localization results for the dual operator exist \cite{JitoKruLiu,JLiu,JLiu_annals} exist, but are very delicate.
\end{enumerate}
Let us now discuss applications of the more precise Theorem \ref{main2} and Corollary \ref{cor_main3}.

\begin{enumerate}
\item In \cite{E92,Amor} it was shown that if $\alpha\in \dc_d$, $v\in C^{\omega}(\T^d;\R)$, and $0<\varepsilon<\varepsilon_0(\alpha,v)$, then the operators $H_x$ with the potential $\varepsilon v$ have purely absolutely continuous spectra (see also \cite{zhao1}), and their Schrodinger cocycles are reducible on a set of energies of full spectral measure. In \cite{gyz}, exponential dynamical localization in expectation (which is stronger that sPDL for all $s$) was established for the corresponding dual operators. Therefore, Corollary \ref{cor_main3} applies. The proof of \cite{gyz} is based on quantitative reducibility estimates. It can be checked that these estimates, actually, guarantee convergence of the integral \eqref{eq_integral_condition} (within a significant margin), and therefore one can also apply Theorem \ref{main2} directly. Therefore, in the setting of \cite{zhao1}, we actually have strong ballistic transport on the whole spectrum, rather than just ballistic motion.
\item A combination of \cite{AFK} and \cite{aj} implies that, if $v\in C^{\omega}(\T;\R)$, $\beta(\alpha)<+\infty$, and $0<\varepsilon<\varepsilon_0(v,\beta(\alpha))$, then the operators $\{H_x\}_{x\in \T}$ with potentials $\varepsilon v$ have purely absolutely continuous spectrum, and the corresponding Schr\"odinger cocycles are reducible for almost every energy. Exponential dynamical localization for the dual operators has been established in \cite{JitoKru} (as stated, only for the almost Mathieu operator, but the argument easily extends to the general long range case, since it relies on \cite[Theorem 5.1]{aj} which is established for the long range case; see also \cite{gyz} for the Diophantine case). Therefore, again, Corollary \ref{cor_main3} implies strong ballistic transport on the whole $\ell^2(\Z)$. Note that, for $\beta=0$, it gives a non-perturbative version of the result of \cite{zhao1}, also with strong ballistic transport.
\end{enumerate}
\section{On reducibility and localization}
In this section, we will refine some of the results from \cite{jk2} in order to extend them to a local quantitative setting. For a function $f\in \rL^2(\T^d)$, denote the Sobolev norm by
$$
\|f\|_{\rH^s(\T^d)}^2=\sum_{m\in \Z^d}(1+|m|)^{2s}|\hat f (m)|^2,
$$
where $\hat{f}(m)$ are the Fourier coefficients:
$$
f(x)=\sum_{m\in \Z^d}\hat f(m)e^{2\pi m\cdot x}.
$$
We will only consider $s>d/2$, in which case $\rH^s$ is embedded into $C(\T^d)$ and its elements are ordinary continuous functions, rather than equivalence classes. Coincidentally, the same condition is sufficient for $\rH^s$ being an algebra with respect to the pointwise multiplication, which will also be important; see, for example \cite[Theorem 4.39]{Adams}.
\begin{prop}
\label{prop_multiplicative_sobolev}
Let $s>d/2$ and $f,g\in \rH^s(\T^d)$. Then $fg\in \rH^s(\T^d)$, and
$$
\|fg\|_{\rH^s(\T^d)}\le C(d,s)\|f\|_{\rH^s(\T^d)}\|g\|_{\rH^s(\T^d)}.
$$
\end{prop}

Let $v\in \rC(\T^d,\R)$, and $\alpha\in \R^d$ such that $\{1,\alpha_1,\ldots,\alpha_d\}$ are independent over $\Q$. Consider the following quasiperiodic Schr\"odinger cocycle $(\alpha,S_{E-v})$, where
$$
S_{E-v}(x)=\begin{pmatrix}
E-v(x)&-1\\
1&0
\end{pmatrix}.
$$
Adapting the definition from Section 1, we will say that $(\alpha,S_{E-v})$ is $\rH^s$-reducible if there exists $B\in \rH^s(\T^d;\GL(2,\C))$ such that
\beq
\label{eq_reducible_def_sec3}
B(x+\alpha)^{-1}S_{E-v}(x)B(x)=\begin{pmatrix}e^{2\pi i \rho(E)}&0\\0& e^{-2\pi i \rho(E)} \end{pmatrix},\quad \forall x\in \T^d,
\eeq
where $\rho(E)$ is the fibered rotation number of $(\alpha,S_{E-v})$, as defined in Section 1. As a consequence, $\deg B=0$.

\begin{defi}
Let $\KK\subset \R$ be a Borel subset. We will say that $(\alpha,S_{E-v})$ is {\it $\rH^s$-reducible on $\KK$} if there exists and an $\rL^2$-normalized balanced family of conjugating matrix functions $\{B(E)\}_{E\in \KK}$, satisfying \eqref{eq_reducible_def_sec3} for all $E\in \KK$, and the following bound:
\beq
\label{eq_sobolev_K}
\int_{\KK}\|B(E,\cdot)\|^4_{\rH^s(\T^d)}\,d\rho(E)<+\infty
\eeq
\end{defi}
At this moment, we also do not assume any regularity of $B_E$ in the variable $E$. For example, $B(E)$ itself may not be measurable in $E$, as long as there is an upper norm bound by a measurable function satisfying \eqref{eq_sobolev_K}. However, one can obtain the following:
\begin{enumerate}
	\item Assuming that $B(E)$ with above properties exist on $\KK$, one can pick a measurable parametrization of $B(E)$ in $E$.
	\item As in Section 1, let $\cR$ be the set of energies such that $(\alpha,S_{E-v})$ is $\rC^s$-reducible. Then, for a given $c_1>0$, the set of $E\in \cR$ such there exists $B(E)$ satisfying \eqref{eq_reducible_def_sec3} with, say, $\|B(E)\|_{\rC^s(\T^d;\mathrm{SL}(2,\R))}\le c_1$, is measurable.
\end{enumerate}
Claims (1) and (2) can be obtained from the following fact: selecting a $B$ satisfying \eqref{eq_reducible_def_sec3} is the same as selecting two linearly independent Bloch wave solutions of the Schr\"odinger equation which, in turn, are completely determined by their initial data. These solutions determine the values of $B(x)$ on a dense subset of $x$, and therefore contain all information on the regularity of the corresponding Bloch functions, as well as the matrix elements of $B$ (as long as the latter are continuous). One can also independently obtain measurability for almost every $E$ (which is equally good in our case) from the duality arguments below, similarly to \cite{jk2}.

Recall that the rotation number of the Schr\"odinger cocycle $(\alpha,S_{E-v})$ is a continuous non-increasing map $\rho\colon \R\to [0,1/2]$, which maps $\Sigma_{\alpha,v}$ onto $[0,1/2]$. The relation
$$
N(E)=1-2\rho(E)
$$
implies that the pre-image of the Lebesgue measure on $[0,1/2]$ under $\rho$ is half of the density of states measure. For $\theta\in [0,1/2]\setminus(\Z+\alpha\cdot\Z^d)$, denote by $E(\theta)$ the unique value $E\in\Sigma_{\alpha,v}$ such that $\rho(E)=\theta$ (note that the uniqueness is violated at the endpoints of spectral gaps, which correspond to the removed values of $\theta$). Extend it as an even function into $[-1/2,0]$, and then extend it into $\R$ by $1$-periodicity. Denote the resulting function, defined on $\R\setminus(\Z+\alpha\cdot\Z^d)$, by the same symbol $E(\theta)$. Let 
$$
\Theta=(\rho(\KK)\cup(-\rho(\KK)))\setminus (\Z/2+\alpha\cdot\Z^d/2).
$$
Then $E\colon \Theta\to \Sigma_{\alpha,v}$ is a measurable map which takes each of its values twice, and whose range is equal to $\KK$ except, at most, for a countable subset. Note that we only needed to remove $\Z+\alpha\cdot \Z^d$ for the above argument. However, the further construction will require removal of half-$\alpha$-rational frequencies.

Let us recall the definition of the dual operator family.
\beq
\label{eq_dual_sec4}
(L_{\theta}\psi) (n)=\sum_{m\in \Z^d}\hat{v}_{n-m}\psi(m)+2\cos 2\pi (n\cdot\alpha+\theta)\psi(n), \quad \theta\in \T.
\eeq
In order to formulate the main result of this section, introduce the translation operator:
$$
T^a\colon \ell^2(\Z^d)\to \ell^2(\Z^d),\quad (T^a\psi)(n):=\psi(n+a).
$$
An important property of the eigenvectors of the operators \eqref{eq_dual_sec4} is the following covariance relation. Suppose that $L_{\theta}\psi=E\psi$, $\psi\in \ell^2(\Z^d)$  Then
\beq
\label{eq_covariance}
L_{\theta+\ell\cdot\alpha}T^\ell\psi(\theta)=E(\theta)T^\ell\psi(\theta),\quad \forall \ell\in \Z^d.
\eeq
As a consequence, if one wants to study localization properties of the family $\{L_{\theta}\}_{\theta\in\T}$, it may be beneficial to pick only one representative from each ``equivalence class'' defined by \eqref{eq_covariance}. There are obvious difficulties with this approach, as it dangerously resembles the procedure of constructing a non-measurable subset of the circle. However, in our setting it is possible and is discussed, for example, in \cite{jk2}. The main result of this section is the following refinement of \cite{jk2}.
\begin{theorem}
\label{main_localization}
Suppose that the family of Schr\"odinger cocycles $(\alpha,S_{E-v})$ is $\rH^s$-reducible on $\KK\subset \R$ with sone $s>d/2$. Construct the subset $\Theta\subset[0,1]$ and the function $E(\cdot)$ as above. Then there exists a measurable function $f\colon \Theta\times \T^d\to\C$, such that the following claims hold.
\begin{enumerate}
\item For each $\theta\in \Theta$, $\|f(\theta,\cdot)\|_{\rL^2(\T^d)}=1$, and
$$
\int_{\Theta}\|f(\theta,\cdot)\|_{\rH^s(\T^d)}^4\,d\theta<+\infty.
$$
\item For each $\theta\in \Theta$, the vector $\psi(\theta;m)=\hat{f}(\theta,m)$ $($the Fourier transform in the second variable$)$ is an eigenvector of the dual operator:
$$
L_{\theta}\psi(\theta)=E(\theta)\psi(\theta).
$$
\item For $\ell\in \Z^d$ such that $\theta-\ell\cdot\alpha\in \Theta$, construct additional eigenvectors of $L_{\theta}$ by
$$
\psi_{\ell}(\theta)=T^\ell\psi(\theta-\ell\cdot\alpha),
$$
so that, using \eqref{eq_covariance}, we have
$$
L_{\theta+\ell\cdot\alpha}\psi_\ell(\theta)=E(\theta)\psi_\ell(\theta).
$$
Then, for almost every $\theta\in \T$, the spectrum of $L_{\theta}(\KK)$ is purely point, and the constructed eigenfunctions
\beq
\label{eq_eigenfunctions}
\{\psi_\ell(\theta)\colon \theta-\ell\cdot\alpha\in \Theta\},
\eeq
form a complete system for $L_{\theta}(\KK)$.
\item Denote by $\psi_{*}(\theta)$ the following convolution vector:
$$
\psi_{*}(\theta;p)=\sum_{m\in \Z^d}|\psi(\theta;m)\psi(\theta;m+p)|.
$$
Then the following Sobolev localization property holds:
\beq
\label{eq_sobolev_localization}
\sum_p(1+|p|)^{2s}\int_{\Theta}|\psi_{\star}(\theta;p)|^2\,d\theta<+\infty.
\eeq
\end{enumerate}
\end{theorem}
\begin{proof}
Most of the the argument is very close to \cite{jk2}. See also a similar argument in \cite[Section 3]{gy}. Let $E(\theta)$ be constructed as above. The arguments of \cite{jk2} imply that one can take 
$$
f(x,\theta)=\frac{B_{11}(x,E(\theta))}{\|B_{11}(x,E(\theta))\|_{\rL^2(\T^d)}}
$$ 
for $\theta\in \Theta\cap[0,1/2]$ and extend it by the relation $f(x,-\theta)=\overline{f(x,\theta)}$ into $\Theta$.  Then, for each $\theta\in \Theta$, $\psi(\theta;n)=\hat{f}(\theta;n)$ would be an $\ell^2(\Z^d)$-normalized eigenfunction of $L_{\theta}$:
\beq
\label{eq_eigenvalue_eq}
L_{\theta}\psi(\theta)=E(\theta)\psi(\theta),
\eeq
which implies the first two claims. Let us establish completeness. Again, the argument is similar to \cite{jk2}: we calculate the ``partial density of states measure'', using the expression \eqref{eq_ids_average} with the spectral projection of $L_{\theta}$ replaced by the projection onto the subspace spanned by eigenfunctions \eqref{eq_eigenfunctions}. If that measure coincides with the complete IDS, this would indicate completeness of the eigenfunctions (for almost every $\theta$). The calculation is straightforward if we assume $\psi(\theta)$ to be measurable. One can recover measurability from that of $B(E)$, but there is also a more direct argument as follows.

Let $P_\ell(\theta)$ be the spectral projection of $L_{\theta}$ onto the eigenspace corresponding to the eigenvalue $E(\theta-\ell\cdot\alpha)$, for $\theta-\ell\cdot\alpha\in \Theta$. The above construction implies $P_\ell(\theta)\neq 0$ for $\theta\in \Theta+\ell\cdot\alpha$. Let $\phi(\theta)$ be a measurable branch of an element from $P_0(\theta)$, $\|\phi(\theta)\|=1$. Eventually, we will show that $\phi(\theta)=c(\theta)\psi(\theta)$ for almost every $\theta$. However, at this point we cannot state that the spectrum of $L_{\theta}$ is simple. Fortunately, for the following calculations $\phi(\theta)$ is just as good as $\psi(\theta)$. Denote
$$
\phi_\ell(\theta):=T^\ell \phi(\theta-\ell\cdot\alpha).
$$
Then, by covariance, we have the following eigenvalue equation similar to \eqref{eq_eigenvalue_eq}.
$$
L_{\theta}\phi_\ell(\theta)=E(\theta-\ell\alpha)\phi_\ell(\theta),\,\,\,\textrm{ if }\,\, \theta-\ell\cdot\alpha\in \Theta.
$$
As a consequence, we have  $\phi_\ell(\theta)\in \Ran P_\ell(\theta)$, and
\begin{multline*}
\sum_{\ell}\int_{\T}\<P_\ell(\theta)\delta_0,\delta_0\>\one_{\Theta}(\theta-\ell\cdot\alpha)\,d\theta\ge 
\sum_{\ell}\int_{\T}|\<\phi_\ell(\theta),\delta_0\>|^2\one_{\Theta}(\theta-\ell\cdot\alpha)\,d\theta\\
=\sum_\ell\int_{\T}|\<\phi(\theta-\ell\cdot\alpha),\delta_{-\ell}\>|^2\one_{\Theta_0}(\theta-\ell\cdot\alpha)\,d\theta
=\sum_\ell\int_{\T}|\<\phi(\theta),\delta_{-\ell}\>|^2\one_{\Theta}(\theta)\,d\theta=|\Theta|=N(\KK).
\end{multline*}
Since the left hand side cannot be larger than $|N(\KK)|$, all inequalities are actually equalities, which also implies simplicity of the spectrum for almost every $\theta$. Since measurable parametrization of eigenvectors was obtained independently of measurability of $B(E)$ and that the eigenvalues of $L_{\theta}$ are simple on $\Theta$, this gives us measurability of $B(E)$ in retrospective.

It remains to establish Claim (4). We will obtain it as a consequence of Claim 1. Let 
$$
f_1(\theta,x)=\sum_{m\in \Z^d}|\psi(\theta;-m)|e^{2\pi i m\cdot x},\quad f_2(\theta,x)=\sum_{m\in \Z^d}|\psi(\theta;m)|e^{2\pi i m\cdot x}.
$$
Clearly, we have
$$
\|f_1(\theta,\cdot)\|_{\rH^s(\T^d)}=\|f_2(\theta,\cdot)\|_{\rH^s(\T^d)}=\|f(\theta,\cdot)\|_{\rH^s(\T^d)}.
$$
Then one can express $\psi_{*}$ as a convolution:
$$
\psi_*(\theta;p)=(\hat f_1(\theta,\cdot)*\hat f_2(\theta,\cdot))(p),
$$
and hence, by definition of the Sobolev norm and Proposition \ref{prop_multiplicative_sobolev}, we have
$$
\sum_{p}(1+|p|)^{2s}|\psi_*(\theta;p)|^2=\|f_1(\theta,\cdot) f_2(\theta,\cdot)\|_{\rH^s}^2\le \|f(\theta,\cdot)\|_{\rH^s}^4.
$$
One can now get Claim (4) by integrating in $\theta$.
\end{proof}
We will also need a Sobolev version of the dynamical localization.
\begin{theorem}
\label{th_sobolev}
Under the assumptions of Theorem $\ref{main_localization}$, there exists $h\in \rH^s(\T^d)$ such that
$$
\int_{\T}|\<\delta_p,\one_{\KK}(L_{\theta})e^{-i t L_{\theta}}\delta_q\>|\,d\theta=\hat h(q-p).
$$
\begin{proof}
We have, using the notation of the previous theorem,
\begin{multline}
\label{eq_covariance_calculation}
h_{pq}:=\int_{\T}|\<\delta_p,\one_{\KK}(L_{\theta})e^{-i t L_{\theta}}\delta_q\>|\,d\theta\le \sum_{\ell\in \Z^d}\int_{\Theta+\ell\cdot\alpha}|\psi_{\ell}(\theta,p)\psi_{\ell}(\theta,q)|\,d\theta\\
=\sum_{\ell\in \Z^d}\int_{\Theta+\ell\cdot\alpha}|\psi(\theta-\ell\cdot\alpha,p+\ell)\psi(\theta-\ell\cdot\alpha,q+\ell)|\,d\theta=\sum_{\ell\in \Z^d}\int_{\Theta}|\psi(\theta,p+\ell)\psi(\theta,q+\ell)|\,d\theta\\
\le \l(\int_{\Theta}|\psi_*(\theta,q-p)|^2\,d\theta\r)^{1/2}.
\end{multline}
Due to covariance, $h_{pq}$ only depends on $q-p$. We have the following thanks to the last claim in Theorem \ref{main_localization}:
$$
\sum_{p\in \Z^d}(1+|p|)^{2s}|h_{0p}|^2\le \int_{\Theta}\sum_{p\in \Z^d}(1+|p|)^{2s}|\psi_*(\theta;p)|^2\,d\theta\le \int_{\Theta}\|f(\theta;\cdot)\|^4_{\rH^s}.\,\qedhere
$$
\end{proof}
\end{theorem}
Suppose that, instead of a Sobolev bound, we have a uniform bound
$$
\sup_{\theta\in \Theta}\|f(\theta;\cdot)\|_{\rC^s(\T^d)}<+\infty.
$$
In this case, the dual operator family demonstrates an extremely strong form of uniform localization, which would allow us, ultimately, to relax regularity requirements on the reducibility.
\begin{lemma}
\label{lemma_uniform_localization}
Suppose that, in the notation of Theorem $\ref{main_localization}$, we have
$$
\sup_{\theta\in \Theta}\|f(\theta;\cdot)\|_{\rC^s(\T^d)}=:M<+\infty.
$$
Then, for almost every $\theta\in \T$, we have the following uniform dynamical localization bound:
$$
\esssup_{\theta\in \T}|\<\delta_p,\one_{\KK}e^{itL_{\theta}}\delta_q\>|<\frac{C(s,M)}{(1+|p-q|)^{2s-d}}.
$$
\begin{proof}
Using the representation from Theorem \ref{main_localization}, we have
$$
|\psi_{\ell}(\theta;q)|=|\hat f(\theta-\ell\cdot\alpha;q+\ell)|\le \frac{C(M)}{(1+|q+\ell|)^s}.
$$
The rest follows from Lemma \ref{lemma_convolution}.
\end{proof}
\end{lemma}

\section{From localization to strong ballistic transport}
In this section, we will prove Theorem \ref{main2} by studying the consequences of the results from the previous section to the operator \eqref{h_def}:
$$
(H_{x}\psi)(n)=\psi(n-1)+\psi(n+1)+v(x+n\alpha)\psi(n),\ \ n\in\Z.
$$
In order to formulate the main result, we will need to introduce the dual operator family. Define the Fourier coefficients of $v$ by $\hat{v}_n$, where
$$
v(x)=\sum_{n\in \Z^d} \hat{v}_n e^{2\pi i n\cdot x}.
$$
Let $\{L_{\theta}\}_{\theta\in \T}$ be the dual family on $\ell^2(\Z^d)$:
\beq
\label{eq_L_def}
(L_{\theta}\psi) (n)=\sum_{m\in \Z^d}\hat{v}_{n-m}\psi(m)+2\cos 2\pi (n\cdot\alpha+\theta)\psi(n).
\eeq
As stated in the introduction, denote by $A$ the current operator on $\ell^2(\Z)$:
$$
(A\psi)(n)=i(\psi(n+1)-\psi(n-1)).
$$
For a Borel subset $\KK\subset\R$, we defined
$$
A(x,\KK)=\one_{\KK}(H_x)A\one_{\KK}(H_x).
$$
Recall also that, by definition,
$$
H_x(\KK)=\l.\one_{\KK}(H_x)\r|_{\Ran \one_{\KK}(H_x)}.
$$
It is convenient to assume that $A(x,\KK)$ acts on the whole $\ell^2(\Z)$ and $H_x(\KK)$ is restricted to 
$\Ran \one_{\KK}(H_x)$, since, in the latter case, the wording ``$\sigma(H_x(\KK))$ is purely absolutely continuous'' has intended meaning and does not need to account for the large kernel of $\Ran \one_{\KK}(H_x)^{\perp}$. Recall the definition of the function $g_{\KK}(E):$
$$
g_{\KK}(E)=\begin{cases}\frac{1}{\pi N'(E)},&\quad E\in \KK\\
0,&E\in \R\backslash\KK.	
\end{cases}
$$
\begin{defi}
An analytic quasiperiodic operator family $\{H_x\}_{x\in \T^d}$ will called $\KK$-{\it regular} if the following properties are satisfied:
\begin{enumerate}
	\item The spectra of $H_x(\KK)$ are purely absolutely continuous.
	\item The families $\one_{\KK}(H_x)$ and $g_{\KK}(H_x)$ are strongly continuous in the parameter $x\in \T^d$.
\end{enumerate}	
\end{defi}
The results of \cite{Deift} imply that, under the above assumptions, $\|g_{\KK}(H_x)\|\le 2$.

\begin{theorem}
\label{th_main_abstract}
Let $\{H_x(\KK)\}_{x\in \T^d}$ be a $\KK$-regular family such that the dual operator family $\{L_{\theta}\}_{\theta\in \T}$ satisfies $\rH^s$-localization in expectation on $\KK$ for some $s>5d/2$, or $s$-uniform power law localization in expectation on $\KK$ for some $s>d$. Then the conclusion of Theorem $\ref{main}$ holds. In other words, for every $x\in \T^d$ the limit
$$
Q(x,\KK)=\slim\limits_{T\to +\infty}\frac{1}{T}\int_0^T e^{itH_x}\one_{\KK}(H_x)A\one_{\KK}(H_x)e^{-itH_x}\,dt,
$$
exists and
$$
\ker{Q(x,\KK)}=\l(\Ran\one_{\KK}(H_x)\r)^{\perp}.
$$
\end{theorem}
\begin{remark}
In Section 6, we obtain $\KK$-regularity as a consequence of local $\rC^1$ rotations reducibility for the corresponding Schr\"odinger cocycles. Therefore, it holds in all considered cases.
\end{remark}
In order to prove Theorem \ref{th_main_abstract}, we will need several additional calculations with duality involving direct integrals. Each of the families $\{H_x\}_{x\in \T^d}$ and $\{L_{\theta}\}_{\theta\in \T}$ can be considered as a single operator in the appropriate direct integral space:
$$
\mathfrak{H}:= \int_{\T^d}^{\oplus} \ell^2(\Z)\,dx,\quad \widetilde{\mathfrak{H}}= {\int_{\T}^{\oplus}\ell^2(\Z^d)}\,d\theta.
$$
Denote the unitary duality operator $\mathcal U\colon \mathfrak{H}\to \widetilde{\mathfrak{H}}$  on functions $\Psi=\Psi(x,n)$ by
\begin{equation}\label{u_def}
(\mathcal {U} \Psi)(\theta,m)=\widetilde{\Psi}(\theta+ m\cdot\alpha,m),
\end{equation}
where $\widetilde\Psi$ denotes the Fourier transform in both discrete and continuous variables: 
\beq
\label{eq_dual_fourier_def}
\widetilde {\Psi}(\theta,m)=\sum\limits_{n\in \Z}\,\int_{\T^d}e^{2\pi i n\theta-2\pi i  m\cdot x}\Psi(x,n)\,dx.
\eeq
In the notation, we will always write the continuous variables before discrete variables in the arguments of functions, even when they roles are switched under duality. As mentioned above, the operator families $\{H_x\}_{x\in \T^d}$ and $\{L_{\theta}\}_{\theta\in \T}$ can be represented by direct integrals
$$
\mathcal{H}:=\int_{\T^d}^{\oplus} H_x\,dx,\quad {\mathcal{L}}:=\int_{\T}^{\oplus} L_\theta\,d\theta.
$$
Aubry duality (see, for example, \cite{GJLS}) can be formulated as the unitary equivalence of the above direct integrals:
\begin{equation}\label{Duality}
\mathcal{U} \mathcal H \mathcal U^{-1}={\mathcal L}.
\end{equation}
One can apply duality to other operators and operator families on $\ell^2(\Z)$. For example, the operator family corresponding to the operator $A$ (constant in $x$) has the following dual family:
$$
\mathcal U\l(\int_{\T^d}^{\oplus}A\,dx\r)\mathcal U^{-1}=\int_{\T}^{\oplus}\widetilde{A}(\theta)\,d\theta,
$$
where
\begin{equation}
\label{A_dual}
(\widetilde A(\theta)\psi)(m)=2\sin 2\pi (m\cdot\alpha+\theta)\psi(m),\quad m\in \Z^d.
\end{equation}
Note that an $x$-independent family may become $\theta$-dependent after the duality transformation, and vice versa. For any (Borel) function $f$, we have
\beq
\label{eq_function_directintegral}
\mathcal U f(\mathcal H)\,\mathcal U^{-1}=\mathcal Uf\l(\int_{\T}^{\oplus}H_x\,dx\r)\mathcal U^{-1}=\mathcal U\l(\int_{\T}^{\oplus}f(H_x)\,dx\r)\mathcal U^{-1}=\int_{\T^d}^{\oplus}f(L_{\theta})\,d\theta=f(\mathcal L).
\eeq
For a Borel subset $\KK$, denote 
$$
\widetilde A(\theta,\KK)=\one_{\KK}(L_{\theta})\widetilde{A}(\theta)\one_{\KK}(L_{\theta}).
$$
Then, one can check that $\wt A(\theta,\KK)$ is dual to $A(x,\KK)$:
$$
\mathcal U\frac{1}{T}\int_0^T\l(\int_{\T^d}^{\oplus}e^{i H_xt}A(x,\KK)e^{-i H_x t}\,dx\r)dt\,\,\mathcal U^{-1}=\frac{1}{T}\int_0^T\l(\int_{\T}^{\oplus}e^{i L_{\theta} t}\widetilde{A}(\theta,{\KK})e^{-i L_{\theta} t}\,d\theta\r)\,dt.
$$
The following proposition is, essentially, established in \cite{Kach1} for the case $\KK=\sigma(H_x)$ in a slightly different form. We include most of the proof for the convenience of the reader.
\begin{prop}
\label{prop_qtilde}
Under the assumptions of Theorem $\ref{th_main_abstract}$, denote by $E_k(\theta)$, $\psi_k(\theta)$ the eigenvalues and eigenfunctions of $L_\theta(\KK)$ $($the exact choice of parametrization does not matter$)$. Then, for almost every $\theta\in \T$, the following limit
$$
\widetilde Q(\theta,\KK):=\slim_{T\to +\infty} \frac{1}{T}\int_0^T e^{i t L_{\theta}}\widetilde A(\theta,\KK)e^{-i t L_{\theta}}\,dt
$$
exists and is a diagonal operator in the representation of eigenvectors of $L_{\theta}(\KK)$. More precisely,
\beq
\label{eq_qtilde_formula}
\widetilde Q(\theta,\KK)\psi_k(\theta)=\frac{1}{\pi N'(E_k(\theta))}\psi_k(\theta).
\eeq
As a consequence, for almost every $\theta\in \T$ we have
$$
\widetilde Q(\theta,\KK)=g_{\KK}(L_{\theta}).
$$
\end{prop}
\begin{proof}
We only sketch the main ideas, since most of the argument is contained in \cite{Kach1}. The existence of the limit and the fact that it is diagonal in the basis of eigenvectors of $L_{\theta}$ follows from the following standard calculation:
$$
\frac{1}{T}\int_0^T\l\< e^{i t L_{\theta}}\wt A(\theta,\KK)e^{-i t L_{\theta}}\psi_k(\theta),\psi_{\ell}(\theta)\r\>\,dt=\l(\frac{1}{T}\int_0^T e^{i t (E_{\ell}(\theta)-E_k(\theta))}\,dt\r)\<\wt A(\theta)\psi_k(\theta),\psi_{\ell}(\theta)\>
$$
and the fact that
$$
\l|\frac{1}{T}\int_0^T e^{i t (E_{\ell}(\theta)-E_k(\theta))}\,dt\r|\le 1;\quad \lim\limits_{T\to +\infty}\frac{1}{T}\int_0^T e^{i t (E_{\ell}(\theta)-E_k(\theta))}\,dt=\begin{cases} 1,&E_k(\theta)=E_{\ell}(\theta);\\
	0,&E_k(\theta)\neq E_{\ell}(\theta).
\end{cases}
$$
As a consequence, we obtain
\beq
\label{eq_qdiag_1}
\<\wt Q(\theta,\KK)\psi_k(\theta),\psi_k(\theta)\>=\<\wt A(\theta)\psi_k(\theta),\psi_k(\theta)\>=\sum_{m\in \Z^d}2\sin 2\pi (\theta+m\cdot\alpha) |\psi_k(\theta,m)|^2.
\eeq
In order to establish \eqref{eq_qtilde_formula}, consider the Fourier transforms of the eigenvectors of $L_{\theta}(\KK)$:
$$
f_k(x,\theta)=\sum_{n\in \Z^d}e^{2\pi i n \cdot x}\psi_k(\theta,n),
$$
where $\psi_k(\theta,n)$ is the $n$th component of $\psi_k(\theta)$ (the latter is considered as a vector from $\ell^2(\Z)$). If $\theta\notin \Z+\alpha\cdot\Z^d$, then (see Appendix C of \cite{AJM}, and also Remark 5.1 in \cite{jk2})
\beq
\label{eq_trivial_kernel}
d_k(\theta):=e^{2\pi i \theta}f_k(x,\theta)\overline{f_k(x-\alpha,\theta)}-e^{-2\pi i \theta}\overline{f_k(x,\theta)}f_k(x-\alpha,\theta)\neq 0.
\eeq
By direct calculation and \eqref{eq_qdiag_1}, we also have 
\beq
\label{eq_dk_equality}
d_k(\theta)=\<\wt Q(\theta,\KK)\psi_k(\theta),\psi_k(\theta)\>,
\eeq
which implies that $\ker\wt {\mathcal Q}=\{0\}$. Let
$$
S_{E-v}(x):=\begin{pmatrix}
	E-v(x)&-1\\1&0
\end{pmatrix}
$$
be the Schr\"odinger cocycle, and consider a matrix function $B(x,\theta)$ defined by
$$
B(x,\theta):=\frac{1}{|d_k(\theta)|^{1/2}}\begin{pmatrix}
f_k(x,\theta)&\overline{f_k(x,\theta)}\\
e^{-2\pi i \theta}f_k(x-\alpha,\theta)&e^{2\pi i \theta}\overline{f_k(x-\alpha,\theta)}
\end{pmatrix};
$$
note that the matrix is invertible since $d_k(\theta)\neq 0$. Then
$$
B(x+\alpha,\theta)^{-1}S_{E-v}(x)B(x,\theta)=\begin{pmatrix}
	e^{2\pi i \theta}&0\\ 0 & e^{-2\pi i \theta}
\end{pmatrix}.
$$
Kotani's theory (see the argument in \cite{Kach1} with additional references) implies that there exists a subset $\mathcal E\subset \KK$ of full Lebesgue measure (as a consequence, full spectral measure for each $H_x(\KK)$) such that, if $E_k(\theta)$ is constructed above and $E_k(\theta)\in \mathcal E$, then
$$
d_k(\theta)=\frac{1}{\pi N'(E_k(\theta))}.
$$
Comparing the last equality with \eqref{eq_dk_equality}, both of which hold for almost every $\theta\in\T$, we complete the proof. Note that the function $g$ is only defined Lebesgue almost everywhere on $\KK$. However, for almost every $\theta\in \T$, all eigenvalues $E_k(\theta)$ will be at differentiability points of $N$, and hence $g(L_\theta)$ will be well defined. As a consequence, using \eqref{eq_function_directintegral}, we have
\begin{align}\label{eq_q_direct}
\mathcal{Q}(\mathcal{K})&:=\int_{\T^d}^{\oplus}Q(x,\KK)\,dx=\int_{\T^d}^{\oplus} g_{\KK}(H_x)\,dx\\ \nonumber
&=\mathcal{U}^{-1}\l(\int_{\T}^{\oplus}\widetilde{Q}(\theta,\KK)\,d\theta\r)\mathcal{U}=\mathcal{U}^{-1} \l(\int_{\T}^{\oplus} g_{\KK}(L_\theta)\,dx\r)\, \mathcal{U}:=\mathcal{U}^{-1}\widetilde{\mathcal{Q}}(\mathcal{K})\, \mathcal{U}.\,\qedhere
\end{align}
\end{proof}

Denote by $Q(x,T,\KK)$ the pre-limit expression:
\beq
\label{eq_prelimit}
Q(x,T,\KK):=\frac{1}{T}\int_0^T e^{i t H_x}A(x,\KK)e^{-i t H_x}\,dt ,\quad \mathcal{Q}(T,\KK)=
\int_{\T^d}^{\oplus}Q(x,T,\KK)\,dx.
\eeq
We would like to show the following, for all $p\in \Z$:
\begin{equation}\label{limit_1}
Q(x,T,\KK)\delta_p\to Q(x,\KK)\delta_p,
\end{equation}
where $\{\delta_p\}_{p\in \Z}$ denote the standard basis vectors in $\ell^2(\Z)$.
Let
$$
f_{p,T}(x)=Q(x,T,\KK)\delta_p,\quad f_p(x)=Q(x,\KK)\delta_p.
$$
Denote by $f_{p,T}(x,n)$ and $f_{p}(x,n)$ the $n$-th components of $f_{p,T}(x)$, $f_{p}(x)$ respectively, where $n\in \Z$.
One can treat $f_{p,T}$ and $f_p$ as elements of $\rL^2(\T^d\times \Z)$. Denote also by $\widetilde{f}_{p,T}(\theta,m)$, $\widetilde{f}_p(\theta,m)$ the Fourier transforms of $f_{p,T}$, $f_p$ in both variables defined as in \eqref{eq_dual_fourier_def}:
$$
\widetilde {f}_{p,T}(\theta,m)=\sum\limits_{n\in \Z}\,\int_{\T^d}e^{2\pi i n\theta-2\pi i  m\cdot x}f_{p,T}(x,n)\,dx.
$$
\begin{lemma}\label{lemma_L1}
For any $p\in\Z$, $x\in\T^d$, and $T>0$, we have
\begin{align}
\label{eq_conclusion_211}
\|f_{p,T}(x)-f_{p}(x)\|_{\ell^2(\Z)}^2\leq \int_{\T}\left(\sum\limits_{m\in\Z^d}|\widetilde{f}_{p,T}(\theta,m)-\widetilde{f}_p(\theta,m)|\right)^2d\theta.
\end{align}
\end{lemma}
\begin{proof}
First, let us note that both $x\mapsto f_{p,T}(x)$ and $x\mapsto f_p(x)$ are continuous as maps from $\T^d$ to $\ell^2(\Z)$. In particular, they are continuous component-wise. Denote by $\hat{f}_{p,T}(x,\theta)$ the Fourier transform only in the variable $n$, and same for $\hat{f}_p(x,\theta)$. Using the Parseval's identity, continuity in $x$, and $\ell^1$ bound for the Fourier transform, we have the following:
\begin{multline}
\sup_x\|f_{p,T}(x)-f_{p}(x)\|_{\ell^2(\Z)}^2=\sup_x\sum_{n\in \Z}|f_{p,T}(x,n)-f_{p}(x,n)|^2=\sup_x \int_{\T}|\hat{f}_{p,T}(x,\theta)-\hat{f}_p(x,\theta)|^2d\theta
\\
\le \int_{\T}\l(\esssup_x |\hat f_{p,T}(x,\theta)-\hat f_{p}(x,\theta)|\r)^2\,d\theta\le\int_{\T}\left(\sum\limits_{m\in\Z^d}|\widetilde{f}_{p,T}(\theta,m)-\widetilde{f}_{p}(\theta,m)|\right)^2\,d\theta.
\end{multline}
\end{proof}
\begin{remark}
It is crucial that the left hand side of \eqref{eq_conclusion_211} is continuous in $x$, otherwise we would not have been able to obtain convergence for all $x\in \T^d$, as the right hand side of \eqref{eq_conclusion_211} does not allow to recover any data about measure zero subsets of $\T^d$ in the variable $x$. The said continuity, ultimately, reduces to the assumption of $\KK$-regularity of the family $\{H_x\}_{x\in \T}$.
\end{remark}
\begin{remark}
Let $\mathcal U$ be the duality operator. Then
$$
\widetilde f_{p,T}(\theta,m)=(\mathcal U f_{p,T})(\theta-m\cdot\alpha,m),\quad \widetilde f_{p}(\theta,m)=(\mathcal U f_{p})(\theta-m\cdot\alpha,m).
$$
Therefore, in order to show \eqref{limit_1}, we can apply Lemma \ref{lemma_L1} and reduce it to a convergence statement about the images of $f_{p,T}$ under duality.
\end{remark}
We will use the following notation for the dual pre-limit expressions:
$$
\widetilde Q(\theta,T,\KK):=
\frac{1}{T}\int_0^T e^{i t L_{\theta}}\widetilde A(\theta,\KK)
e^{-i t L_{\theta}}\,dt,\quad \widetilde{\mathcal{Q}}(T,\KK)=
\int_{\T}^{\oplus}\widetilde{Q}(\theta,T,\KK)\,d\theta.
$$
Finally, consider $\delta_p$ as an element of $\rL^2(\T^d\times \Z)$ that is a constant function in the $x$ variable. Then, it's Fourier transform in both variables is equal to
$$
(\widetilde\delta_p)(\theta,m)=(\mathcal U\delta_p)(\theta-m\cdot\alpha,m)=e^{2\pi i p\theta}\delta_0(m),
$$
which implies
\begin{small}
$$
\widetilde f_{p,T}(\theta,m)=(\mathcal U \mathcal Q(T,\KK) \delta_p )(\theta-m\cdot\alpha,m)=(\widetilde{\mathcal Q}(T,\KK)\,\mathcal U{\delta}_p)(\theta-m\cdot\alpha,m)=e^{2\pi i pm\cdot\alpha}\wt Q(\theta-m\cdot\alpha,T,\KK)\delta_0,
$$
\end{small}
and similarly
\begin{small}
$$
\widetilde f_{p}(\theta,m)=(\mathcal U \mathcal Q(\KK) \delta_p )(\theta-m\cdot\alpha,m)=(\widetilde{\mathcal Q}(\KK)\,\mathcal U{\delta}_p)(\theta-m\cdot\alpha,m)=e^{2\pi i pm\cdot\alpha}\wt Q(\theta-m\cdot\alpha,\KK)\delta_0.
$$
\end{small}
As a consequence, we can rewrite the conclusion of Lemma \ref{lemma_L1} as
\begin{multline}
\label{eq_ell2_pre}
\sup_x \|f_{p,T}(x)-f_p(x)\|_{\ell^2(\Z)}\le \l(\int_{\T}\left(\sum\limits_{m\in\Z^d}\l|\widetilde{f}_{p,T}(\theta,m)-\widetilde{f}_{p}(\theta,m)\r|\right)^2\,d\theta\r)^{1/2}
\\=\l(\int_{\T}\l(\sum_{m\in \Z^d}\l|\l(\widetilde Q(\theta-m\cdot\alpha,T,\KK)\delta_0-\widetilde Q(\theta-m\cdot\alpha,\KK)\delta_0\r)(m)\r|\r)^2d\theta\r)^{1/2}\\
\le \sum_{m\in \Z^d}\l(\int_{\T}\l|\l(\widetilde Q(\theta-m\cdot\alpha,T,\KK)\delta_0-\widetilde Q(\theta-m\cdot\alpha,\KK)\delta_0\r)(m)\r|^2 d\theta\r)^{1/2}\\
=\sum_{m\in \Z^d} \l\|\l\<\delta_m,\wt Q(\cdot,T,\KK)\delta_0-\wt Q(\cdot,\KK)\delta_0\r\>\r\|_{\rL^2(\T)}=\l\|\wt{\mathcal Q}(T,\KK)\delta_0-\wt {\mathcal Q}(\KK)\delta_0\r\|_{\ell^1(\Z^d;\rL^2(\T))}.
\end{multline}
The factor $e^{2\pi i p m\cdot\alpha}$ was absorbed into the absolute value, and the second inequality is the triangle inequality. Let
$$
(P_N\Psi)(\theta,m)=\begin{cases} \Psi(\theta,m),&|m|\le N\\
0,&|m|>N
\end{cases}
$$
be the projection onto a neighborhood of the origin in discrete $\Z^d$ variable. The following is the main technical estimate of this section that uses the localization bounds.
\begin{lemma}
Suppose that the family $\{L_{\theta}\}_{\theta\in \T}$ satisfies Sobolev localization on $\KK$ in the sense of Theorem $\ref{th_sobolev}$ with $s>5d/2$. Then the norms
$$
\|\mathcal {Q}(T,\KK)\delta_0\|_{\ell^1(\Z^d;\ell^2(\T))}
$$
are bounded uniformly in $T$. As a consequence,
$$
\|(1-P_N)\widetilde{\mathcal Q}(T,\KK)\delta_0\|_{\ell^1(\Z^d;\rL^2(\T))}\le c(N),
$$
where $c(N)\to 0$ as $N\to+\infty$, uniformly in $T$.
\end{lemma}
\begin{proof}
First, let us replace $\wt Q(\theta,T,\KK)$ by the non-averaged expression $e^{i t L_{\theta}}\widetilde A(\theta,\KK) e^{-i t L_{\theta}}$ (thus proving a stronger inequality). As a consequence, we would like to estimate
\begin{multline}
\label{eq_triangle} 
\sum_{n\in \Z^d}\l(\int_{\T}\l|\<\delta_n,e^{i t L_{\theta}}\widetilde A(\theta,\KK) e^{-i t L_{\theta}}\delta_0\>\r|^2\,d\theta\r)^{\frac{1}{2}}
\le 2\sum_{n\in \Z^d}\l(\int_{\T}\l|\<\delta_n,e^{i t L_{\theta}}\widetilde A(\theta,\KK) e^{-i t L_{\theta}}\delta_0\>\r|\,d\theta\r)^{\frac{1}{2}}\\
\le 2\sum_{n\in \Z^d}\l(\sum_{k\in \Z^d}\int_{\T}\l|\<\one_{\KK}(L_{\theta})e^{-i L_{\theta}t}\delta_n,\delta_k\>\<\wt A(\theta)\delta_k,\one_{\KK}(L_{\theta})e^{-i L_{\theta}t}\delta_{0}\>\r|\,d\theta\r)^{\frac{1}{2}}\\
\le 4\sum_{n\in \Z^d}\l(\sum_{k\in \Z^d}\int_{\T}\l|\<\one_{\KK}(L_{\theta})e^{-i L_{\theta}t}\delta_n,\delta_k\>\<\delta_k,\one_{\KK}(L_{\theta})e^{-i L_{\theta} t}\delta_{0}\>\r|\,d\theta\r)^{\frac{1}{2}},
\end{multline}
where in the second inequality we used the fact that the integrand is bounded by $2$ in absolute value to replace $\rL^2$ norm by $\rL^1$ norm, and then used the fact that $\wt A(\theta)$ is a diagonal operator acting on $\delta_k$ as a scalar (we also transferred $\one_{\KK}(L_{\theta})$ to $e^{-i L_{\theta} t}$, and hence there is no more $\KK$ in $\wt A(\theta)$). 
In the case of sPDL, we can continue the chain of inequalities as follows, using Lemma \ref{lemma_convolution}:
\begin{multline*}
\eqref{eq_triangle}\le 4\sum_{n\in \Z^d}\l(\sum_{k\in \Z^d}\int_{\T}\l|\<\one_{\KK}(L_{\theta})e^{-i L_{\theta}t}\delta_n,\delta_k\>\<\delta_k,\one_{\KK}(L_{\theta})e^{-i L_{\theta} t}\delta_{0}\>\r|\,d\theta\r)^{1/2}\\
\le 4\sum_{n\in \Z^d}\l(\sum_{k\in \Z^d}\int_{\T}|\<\one_{\KK}(L_{\theta})e^{-i L_{\theta}t}\delta_n,\delta_k\>|^{1/2}|\<\delta_k,\one_{\KK}(L_{\theta})e^{-i L_{\theta} t}\delta_{0}\>|^{1/2}\,d\theta\r)^{1/2}\\
\le 4\sum_{n\in \Z^d}\l(\sum_{k\in \Z^d}\l(\int_{\T}|\<\one_{\KK}(L_{\theta})e^{-i L_{\theta}t}\delta_n,\delta_k\>|\,d\theta\int_{\T}|\<\delta_k,\one_{\KK}(L_{\theta})e^{-i L_{\theta} t}\delta_{0}\>|\,d\theta\r)^{1/2}\r)^{1/2}.
\end{multline*}
Let
$$
h_n=\int_{\T}|\<\one_{\KK}(L_{\theta})e^{-i L_{\theta}t}\delta_0,\delta_n\>|\,d\theta,\quad n\in \Z^d.
$$
Then, by covariance, we have the following bound (recall that $*$ denotes the standard convolution for functions on $\Z^d$):
$$
\eqref{eq_triangle}\le 4\sum_{n\in \Z^d\colon |n|>N}\l(h^{1/2}*h^{1/2}\r)^{1/2}(n).
$$
Therefore, in order to obtain decay, we need to verify $\l(h^{1/2}*h^{1/2}\r)^{1/2}\in \ell^1(\Z^d)$. We will need to use some bounds on weighted $\ell^2$ spaces with the norms
$$
\|u\|_{\ell^2_s}^2=\sum_{n\in \Z^d}(1+|n|)^{2s}|u(n)|^2.
$$
Their properties are summarized in the Appendix. Since $h\in \ell^2_s(\Z^d)$, we have the following inclusions, see also Appendix (``$+$'' means the number has to be strictly larger):
$$
h^{1/2}\in \ell^2_{s/2-d/4+};
$$
$$
h^{1/2}*h^{1/2}\in \ell^2_{s-d+};
$$
$$
w:=(h^{1/2}*h^{1/2})^{1/2}\in \ell^2_{s/2-3d/4+};
$$
$$
\{(1+|n|)^{s/2-3d/4}w(n)\}_{n\in \Z^d}\in \ell^2(\Z^d).
$$
In order to get $w$ into $\ell^1(\Z^d)$, we can use H\"older inequality, for which it would be sufficient to have
$$
\{(1+|n|)^{-(s/2-3d/4)}\}_{n\in \Z^d}\in \ell^2(\Z^d).
$$
This, ultimately, gives us the requirement $s/2-3d/4>d/2$, which reduces to $s>5d/2$.
\end{proof}
\begin{cor}
\label{cor_fatou}
The conclusion of Lemma $\ref{lemma_L1}$ also holds for $\wt{\mathcal Q}(\KK)$.
\end{cor}
\begin{proof}
Recall that, being a direct integral, $\wt{\mathcal Q}(T,\KK)$ converges to $\wt{\mathcal Q}(\KK)$ in the strong operator topology on $\rL^2(\T\times\Z^d)$. As a consequence, there is a subsequence of time scales $T_k$ such that $\wt{\mathcal Q}(T_k,\KK)\delta_0$ converges to $\wt{\mathcal Q}(\KK)\delta_0$ almost everywhere on $\T\times\Z^d$ as $k\to \infty$ (here, as before, $\delta_0$ is considered as an element of $\rL^2(\T\times\Z^d)$ constant in $\theta$). Hence, the result follows from Fatou's lemma.
\end{proof}
\subsection*{Conclusion of the proof of Theorem \ref{th_main_abstract}} Since the operator norms of $Q(x,T,\KK)$ are uniformly bounded, it suffices to show that $Q(x,T,\KK)\delta_p\to Q(x,\KK)\delta_p$ for all $p\in \Z$. In other words, it is sufficient to show that the right hand side of \eqref{eq_ell2_pre} converges to zero. Take $N\gg 1$. Using the triangle inequality, Lemma \ref{lemma_L1}, and Corollary \ref{cor_fatou}, we have
\begin{multline*}
\l\|\widetilde{\mathcal Q}(T,\KK)\delta_0-\widetilde{\mathcal Q}(\KK)\delta_0\r\|_{\ell^1(\Z^d;\rL^2(\T))}\\
\le
\l\|P_N\l(\widetilde{\mathcal Q}(T,\KK)\delta_0-\widetilde{\mathcal Q}(\KK)\delta_0\r)\r\|_{\ell^1(\Z^d;\rL^2(\T))}+
\l\|(1-P_N)\l(\widetilde{\mathcal Q}(T,\KK)\delta_0-\widetilde{\mathcal Q}(\KK)\delta_0\r)\r\|_{\ell^1(\Z^d;\rL^2(\T))}\\
\le (2N)^{d/2}\l\|\widetilde{\mathcal Q}(T,\KK)\delta_0-\widetilde{\mathcal Q}(\KK)\delta_0\r\|_{L^{2}(\T\times\Z^d)}+2 c(N),
\end{multline*}
where in the last inequality we use the fact that $\ell^1(\Z^d;\rL^2(\T))$ norm is bounded by the $\rL^2(\T\times\Z^d)$ norm on the range of $P_N$ (with an appropriate constant). Now, since $\|\widetilde{\mathcal Q}(T,\KK)\delta_0-\widetilde{\mathcal Q}(\KK)\delta_0\|_{L^{2}(\T\times\Z^d)}\to 0$, the proof can be completed using the standard $\varepsilon/2$ argument.\,\qed.
\vskip 2mm
\subsection*{The proof of Theorem \ref{th_main_abstract} in the uniform case} Suppose, instead of Theorem \ref{th_sobolev}, we have the conclusion of Lemma \ref{lemma_uniform_localization} with $s>d$. By covariance, we have
$$
h_{n-m}=\esssup_{\theta\in \T}|\<\delta_n,\one_{\KK}(L_{\theta)}e^{itL_{\theta}}\delta_m\>|
$$
for some $h$ which satisfies
$$
h(n)\le M(1+|n|)^{-s}.
$$
Similarly to \eqref{eq_triangle}, we can estimate
$$
\l|\<\delta_n,e^{i t L_{\theta}}\widetilde A(\theta,\KK) e^{-i t L_{\theta}}\delta_0\>\r|\le (h*h)(n)\le \frac{C(M)}{(1+|n|)^{2s-d}}.
$$
Therefore, the conclusion reduces to $\{(1+|n|)^{-(2s-d)}\}_{n\in \Z^d}\in \ell^1(\Z^d)$, which is satisfied for $s>d$.
\subsection{On the proofs of the main results} The proof of Theorem \ref{main2} is complete, modulo $\KK$-regularity which will be established in the next section. To finish the proof of Theorem \ref{main}, consider
$$
\KK_j=\{E\in \KK\colon \|B(E)\|_{\rC^s(\T^d)}\le j\}
$$
and apply the uniform result on each $\KK_j$, together with $\KK$-regularity. To prove Corollary \ref{cor_main3}, follow the same lines as in the proof of Theorem \ref{th_main_abstract}, using Lemma \ref{lemma_convolution} instead of the Sobolev bounds. We did not try to optimize the condition $s>4d$ in this case.
\section{Regularity of the absolutely continuous spectral measures}
This section is mostly expository. Let $\{H_x\}_{x\in \T^d}$ be a quasiperiodic operator family \eqref{h_def}:
$$
(H_{x}\psi)(n)=\psi(n-1)+\psi(n+1)+v(x+n\alpha)\psi(n),\ \ n\in\Z.
$$
For a Borel subset $\KK\subset\R$, we will say that the family of Schr\"odinger cocycles $(\alpha,S_{E-v})$ is $\rC^s$-uniformly rotations reducible on $\KK$, if there exists $c>0$ and a family of matrices $B(E;\cdot)\in \rC^s(\T^d;\SL(2,\R))$, $E\in \KK$, such that
\beq
\label{eq_rotation_red}
B(x+\alpha,E)^{-1}S_{E-v}(x)B(x,E)=R_{\theta(x)}=
\begin{pmatrix}
\cos2 \pi\theta(x) & -\sin2\pi\theta(x)\\
\sin2\pi\theta(x) & \cos2\pi\theta(x)
\end{pmatrix}.
\eeq
where $\theta\in \rC^s(\T;\R)$ and
$$
\|B(\cdot,E)\|_{\rC^s(\T^d;\SL(2,\R))}\le c,\quad \forall E\in \KK.
$$
For a Borel subset $F\subset \R$, denote by
$$
\mu^x_{pq}(F)=\<\delta_p,\one_F(H_x)\delta_q\>
$$
a spectral measure of $H_x$. Clearly, $\mu_{pq}^x=\mu_{p+n,q+n}^{x+n\alpha}$. Hence, one simplify the computations by assuming $p=0$.

Moreover, since $\delta_0$ and $\delta_1$ form a cyclic subspace for $H_x$, one can easily check the following (say, by repeatedly applying $H_x$ to $\delta_0$ or $\delta_1$ and eliminating previous elements by induction):
$$
\delta_n=p_x^{(n-1)}(H_x)\delta_0+q_x^{(n-1)}(H_x)\delta_1,
$$
where $p_x^{(n-1)}$, $q_x^{(n-1)}$ are polynomials of degree $\le n-1$, whose coefficients are $\rC^s$-smooth in $x$. As a consequence, in order to establish smoothness of spectral measures (see below), it would sufficient to consider $\mu_{00}^x$ and $\mu_{01}^x(x)$.
\begin{prop}
\label{prop_uc}Suppose that a family of Schr\"odinger cocycles $(\alpha,S_{E-v})$ is $\rC^1$-uniformly rotations reducible on a Borel subset $\KK\subset \R$. Then, for all $x\in\T^d$, and any Borel subset $F\subset \KK$, we have
\beq
\label{eq_mu00}
\mu_{00}^x(F)=\frac{1}{2\pi}\int_F \l(b_{21}(x,E)^2+b_{22}(x,E)^2\r)\,dE.
\eeq
\beq
\label{eq_mu11}
\mu_{01}^x(F)=\frac{1}{2\pi}\int_F \l(b_{21}(x,E)  b_{21}(x+\alpha,E)+b_{22}(x,E) b_{22}(x+\alpha,E)\r)\,dE.
\eeq
\end{prop}
\begin{proof}
Both claims follow from some standard calculations from the Kotani theory. We will, essentially, use the notation from \cite{Damanik}. For $\im E>0$ and $x\in \T^d$, denote by $u_{\pm}(x,E)$ the unique solutions of the eigenvalue equation $H u =Eu$ satisfying
$$
u_{\pm}(x,E;0)=1,\quad u_{\pm}(x,E;n)\to 0\,\,\text{ as }\,n\to \pm\infty,
$$
and the $m$-functions
$$
m_{\pm}(x,E)=-u_{\pm}(x,E;\pm 1).
$$
Using the eigenvalue equation, one can also obtain
$$
u_-(x,E;1)=m_-(x,E)+E-v(x).
$$
The Green's function can be expressed through the above Jost solutions:
\begin{small}
$$
G_{nm}(x,E)=\<\delta_n,(H_x-E)^{-1}\delta_m\>=-\frac{u_-(x,E;n)u_+(x,E;m)}{m_+(x,E)+m_-(x,E)+E-v(x)}.
$$
\end{small}
As a consequence,
\begin{small}
\beq
\label{eq_green}
G_{00}(x,E)=\frac{-1}{m_+(x,E)+m_-(x,E)+E-v(x)},\quad G_{01}(x,E)=\frac{m_+(x,E)}{m_+(x,E)+m_-(x,E)+E-v(x)}.
\eeq
\end{small}
We can also extend $m_{\pm}(E,x)$ into $E\in \R$ by considering limits $m_{\pm}(E+i\varepsilon,x)$ as $\varepsilon\to 0+$, which will exist for almost every $E$ for which $L(E)=0$. In particular, they will exist almost everywhere on $\KK$. The values of $m_{\pm}(E,x)$ for $E\in \R$ can be calculated as follows. Any matrix $B\in \SL(2,\R)$ defines the following action on the upper half plane $\C^+$:
$$
B\circ z=\frac{B_{11}z+B_{12}}{B_{21}z+B_{22}},\quad z\in\C^+.
$$
Suppose that $B(\cdot,E)$ satisfies \eqref{eq_rotation_red}. Then, for almost every pair $(E,x)\in \KK\times\T^d$, we have 
\beq
\label{eq_mpm_b}
m_+(x,E+i0)=B(x,E)\circ i=-\overline{m_-(x,E+i0)}.
\eeq
The continuity arguments similar to \cite{avila_mathieu} (see also \cite[Footnote on page 10]{Damanik}) imply that \eqref{eq_mpm_b} actually holds for {\it all} $x\in \T^d$ and almost every $E\in \KK$, where ``almost every'' depends on $x$. However, in the following considerations zero measure sets will not be important, and hence one can use \eqref{eq_mpm_b} as an alternative definition of $m_{\pm}(x,E+i0)$. Using \eqref{eq_mpm_b}, we can calculate
$$
m_+(x,E+i0)=\frac{b_{11} i+b_{12}}{b_{21}i+b_{22}}=\frac{b_{12}b_{22}+b_{11}b_{21}}{b_{21}^2+ b_{22}^2}+i\frac{1}{ b_{21}^2+ b_{22}^2},
$$
where $b_{ij}=b_{ij}(x,E)$ are the matrix elements of $B(x,E)$. Note that \eqref{eq_mpm_b} implies that the denominators in \eqref{eq_green} are purely imaginary for $E\in \KK+i 0$. Therefore, one can calculate densities of spectral measures $\mu_{00}^x$, $\mu_{01}^x$ as follows:
$$
\frac{d\mu_{00}^x}{dE}=\frac{1}{\pi}\im G_{00}(x,E+i0)=\frac{1}{2\pi \im m_+(x,E+i0)}=\frac{b_{21}^2+b_{22}^2}{2\pi}.
$$
$$
\frac{d\mu_{01}^x}{dE}=\frac{1}{\pi}\im G_{01}(x,E+i0)=-\frac{\re\, m_+(x,E+i0)}{2\pi \im m_+(x,E+i0)}=-\frac{b_{12}b_{22}+b_{11}b_{21}}{2\pi}.
$$
\end{proof}
We immediately obtain the following regularity claim.
\begin{cor}
\label{cor_uniform_lipschitz}
Under the assumptions of Proposition $\ref{prop_uc}$, the spectral measures $\mu^x_{pq}$ are absolutely continuous on $\KK$ with respect to the Lebesgue measure. Moreover, their densities are Lipschitz continuous in $x$:
$$
\l|\frac{d\mu_{pq}^x}{dE}-\frac{d\mu_{pq}^y}{dE}\r|\le C_{p-q}|x-y|,
$$
where $C_{p-q}$ depends on $p-q$ and the constant $c$ from the uniform rotations reducibility assumption.
\end{cor}
\begin{theorem}
\label{th_regular}
Suppose that the family $\{H_x\}_{x\in \T^d}$ is $\rC^1$-uniformly rotations reducible on a Borel subset $\KK\subset \R$. Let $g\in \rL^{\infty}(\R)$, $\supp g\subset \KK$. Then, for any $x_0\in \T^d$, we have
$$
\slim\limits_{x\to x_0} g(H_x)=g(H_{x_0}).
$$
\end{theorem}
\label{th_strong_g}
\begin{proof}
Since $g(H_x)$ are uniformly bounded, it would be sufficient to show $g(H_x)\delta_n\to g(H_{x_0})\delta_n$ strongly. By shifting the $x$ variable, one can assume $n=0$. Since $g(H_x)\delta_0$ are also uniformly bounded in $\ell^2(\Z)$, it is sufficient to establish the following:
\beq
\label{eq_weak1}
\<\delta_n,g(H_x)\delta_0\>\to \<\delta_n,g(H_{x_0})\delta_0\>,\quad \forall n\in \Z
\eeq
\beq
\label{eq_weak2}
\|g(H_x)\delta_0\|\to \|g(H_{x_0})\delta_0\|.
\eeq
To establish \eqref{eq_weak1}, note
$$
\l|\<\delta_n,g(H_x)\delta_0\>-\<\delta_n,g(H_{x_0})\delta_0\>\r|=\l|\int g(E)d\mu^{x}_{n0}(E)-\int g(E)d\mu^{x_0}_{n0}(E)\r|\le \|g\|_{\rL^{\infty}}c_{n}|x-x_0||\KK|,
$$
where $c_n$ is the constant from Corollary \ref{cor_uniform_lipschitz}. Similarly, \eqref{eq_weak2} can be established using the fact
$$
\|g(H_x)\delta_0\|^2=\<\delta_0,|g(H_x)|^2\delta_0\>,
$$
and then repeating the earlier argument applied to the function $|g|^2$.
\end{proof}
\begin{cor}
Suppose that the conclusion of Theorem $\ref{th_strong_g}$ is satisfied for a fixed function $g\in \rL^{\infty}(\R)$ and a sequence of Borel subsets $\KK_1\subset\KK_2\ldots$. Then it also satisfied for $\KK=\cup_j \KK_j$.
\end{cor}
\begin{proof}
The statement follows from the Banach --- Steinhaus theorem: indeed, the family of operators $\{g(H_x)\}_{x\in \T^d}$ is uniformly bounded, and the convergence can be verified on the dense set $\cup_j \Ran(\one_{\KK_j}(H_x))$, by applying the previous theorem with $\KK=\KK_j$.
\end{proof}
\section{Appendix}
In this section, we will establish some elementary bounds which will happen to be useful later. All functional spaces denoted by $\ell$ with some indices will be on $\Z^d$. 
\noindent Denote by $\ell^2_s$ the space of Fourier transform of functions from $\rH^s(\T^d)$ with the following norm:
$$
\|u\|_{\ell^2_s}^2=\sum_{n\in \Z^d}(1+|n|)^{2s}|u(n)|^2.
$$
Recall the H\"older inequality: for $u\in \ell^p$, $v\in \ell^q$, we have
\beq
\label{eq_holder}
\|uv\|_{\ell^r}\le \|u\|_{\ell^p}\|v\|_{\ell^q},\quad \frac{1}{r}=\frac{1}{p}+\frac{1}{q},\quad 1\le p,q,r\le\infty.
\eeq
The following lemma is elementary:
\begin{lemma}
\label{lemma_convolution}
Let $a\in \Z^d$ and $s_1+s_2-d>0$.
$$
\sum_{n\in \Z^d}\frac{1}{(1+|a-n|)^{s_1}(1+|n|)^{s_2}}\leq \frac{c(s_1,s_2,s_3,d)}{(1+|a|)^{s_1+s_2-d}}.
$$
\end{lemma}
\begin{proof}
We have
\begin{small}
\begin{multline*}
\sum_{n\in \Z^d}\frac{1}{(1+|a-n|)^{s_1}(1+|n|)^{s_2}}\le \l(\sum_{|n|\le a/2}+\sum_{|n-a|\le a/2}+\sum_{|n|>a/2,\,|n-a|>a/2}\r)\frac{1}{(1+|a-n|)^{s_1}(1+|n|)^{s_2}}\\
\le \frac{1}{(1+|a/2|)^{s_1}}\sum_{|n|\le a/2}\frac{1}{(1+|n|)^{s_2}}+\frac{1}{(1+|a/2|)^{s_2}}\sum_{|n|\le a/2}\frac{1}{(1+|n|)^{s_1}}+\sum_{|n|\ge  a/4}\frac{1}{(1+|n|/4)^{s_1+s_2}}\\
\le c(s_1,s_2,a)(1+|a|)^{d-s_1-s_2}.
\end{multline*}
\end{small}
\end{proof}
\noindent Finally, recall that $u(p)=(1+|n|)^{-s}$ belongs to $\ell^{r}$ for $rs>d$. We will need the following ``square root'' bound.
\begin{lemma}
Suppose that $u\in \ell^2_s$, and let $v(n)=|u(n)|^{1/2}$. Then
$$
\|v\|_{\ell^2_r}\le C(r,s,d)\|u\|_{\ell^2_s},\quad 0\le r<\frac{s}{2}-\frac{d}{4}.
$$
\end{lemma}
\begin{proof}
The condition $u\in \ell^2_s$ is equivalent to $\{(1+|n|)^s u(n)\}_{n\in \Z^d}\in \ell^2$, which is in turn equivalent to $\{(1+|n|)^{s/2}v(n)\}_{n\in \Z^d}\in \ell^4$. We can multiply it by an appropriate power of $(1+|n|)^{-1}$ in order to get it back to $\ell^2$: since $\{(1+|n|)^{s'}\}_{n\in \Z^d}\in \ell^4$ for $s'>d/4$, we have by H\"older inequality
$$
\{(1+|n|)^{-s'}(1+|n|)^{s/2}v(n)\}_{n\in \Z^d}\in \ell^2,\quad s'>d/4.
$$
which implies the statement of the lemma.
\end{proof}
Finally, the following is the dual version of the multiplicative Sobolev inequality \cite[Theorem 4.39]{Adams} on the language of convolutions. One can also prove it directly and use in the proof of multiplicative inequalities.
\begin{lemma}
Let $u\in \ell^2_{s_1}$, $v\in \ell^2_{s_2}$. Denote their convolution by
$$
(u*v)(n)=\sum_{m\in\Z^d}u(n-m)v(m).
$$
Then
$$
\|u*v\|_{\ell^2_s}\le C(s,s_1,s_2)\|u\|_{\ell^2_{s_1}}\|v\|_{\ell^2_{s_2}},\quad 0<s<s_1+s_2-d/2.
$$
\end{lemma}

\end{document}